\documentclass[a4paper,reqno,11pt]{article}

\usepackage[hmargin=2cm,vmargin=2cm]{geometry}

\usepackage{amsmath,amssymb,amsthm,mathrsfs,mathdots,color,eucal,authblk,enumerate,framed}

\usepackage[colorlinks=true, pdfstartview=FitV, urlcolor=blue, citecolor=red, linkcolor=blue]{hyperref}

\usepackage[utf8]{inputenc}
\usepackage[T1]{fontenc}

\def\res{\mathop{\mathrm {res}}\limits_}
\def\Z{\mathbb{Z}}
\def\C{\mathbb{C}}
\def\R{\mathbb{R}}

\def\J{\mathrm{J}}
\def\wt{\widetilde}
\def\wh{\widehat}
\def\G{\Gamma}
\def\d{\mathrm d}
\def\e{\mathrm{e}}
\def\i{\mathrm{i}}
\def\pa{\partial}
\def\bs{\boldsymbol}
\def\s{\sigma}
\def\a{\alpha}
\def\b{\beta}
\def\g{\gamma}
\def\tr{\mathrm{tr}\,}
\newcommand{\op}[1]{\mathcal{#1}}
\def\Be{\mathsf{Be}}
\def\reg{\mathsf{reg}}
\def\Ai{\mathrm{Ai}}

\def\be{\begin{equation}}
\def\ee{\end{equation}}

\newtheorem{theorem}{Theorem}[section]

\newtheorem{lemma}[theorem]{Lemma}
\newtheorem{proposition}[theorem]{Proposition} 
\newtheorem{corollary}[theorem]{Corollary}
\newtheorem{remark}[theorem]{Remark}

\newtheorem{theoremintro}{Theorem}

\begin{document}

\numberwithin{equation}{section}

\title{Integrable equations associated with the finite-temperature deformation of the discrete Bessel point process}
\author[1]{Mattia Cafasso}
\author[2]{Giulio Ruzza}
\renewcommand\Affilfont{\small}
\affil[1]{\textit{Univ Angers, SFR MATHSTIC, F-49000 Angers, France;} \texttt{mattia.cafasso@univ-angers.fr}}
\affil[2]{\textit{IRMP, UCLouvain, Chemin du Cyclotron 2, 1348 Louvain-la-Neuve, Belgium;} \texttt{giulio.ruzza@uclouvain.be}}

\date{}
\maketitle

\begin{abstract}
We study the finite-temperature deformation of the discrete Bessel point process.
We show that its largest particle distribution satisfies a reduction of the 2D Toda equation, as well as a discrete version of the integro-differential Painlev\'e II equation of Amir--Corwin--Quastel, and we compute initial conditions for the Poissonization parameter equal to 0.
As proved by Betea and Bouttier, in a suitable continuum limit the last particle distribution converges to that of the finite-temperature Airy point process.
We show that the reduction of the 2D Toda equation reduces to the Korteweg--de\thinspace Vries equation, as well as the discrete integro-differential Painlev\'e II equation reduces to its continuous version.
Our approach is based on the discrete analogue of Its--Izergin--Korepin--Slavnov theory of integrable operators developed by Borodin and Deift.
\end{abstract}

\medskip
\medskip



\medskip

\section{Introduction and results}

\subsection{Finite-temperature discrete Bessel point process and the 2D Toda equation}
In this paper we study the {\it finite-temperature discrete Bessel point process}, which is the determinantal point process on~$\Z': =\Z+\tfrac 12$ with correlation kernel
\be
\label{eq:K}
K_\sigma^\Be(a,b)=\sum_{l \in \mathbb Z'} \s(l)\J_{a + l}(2L)\J_{b + l}(2L),\qquad a,b\in\Z',
\ee
where $L>0$ is a parameter, $\J_k(\cdot)$ is the Bessel function of first kind of order~$k$, and $\s:\Z'\to[0,1]$ is a function such that $\sigma\in\ell^1(\Z'\cap (-\infty,0))$.
The fact that the kernel~\eqref{eq:K} actually induces a determinantal point process on~$\Z'$ and the role of the decay conditions on~$\s$ at~$-\infty$ will be clarified in Section~\ref{sec:2}. 

\medskip

The specialization~$\s=1_{\Z'_+}$ of~\eqref{eq:K}, where $\Z'_+:=\Z'\cap(0,+\infty)$, yields the standard discrete Bessel point process~\cite{BOO, JohanssonPlan}, namely, the determinantal point process with correlation kernel
\be
\label{eq:kernelintegrable}
K^\Be(a,b)=\sum_{l \in \mathbb Z'_+}\J_{a + l}(2L)\J_{b + l}(2L)=L\,\frac{\J_{a-\frac 12}(2L)\J_{b+\frac 12}(2L)-\J_{a+\frac 12}(2L)\J_{b-\frac 12}(2L)}{a-b},\quad a,b\in\Z'.
\ee
(The last equality easily follows from a property of the Bessel functions and will be proved for the reader's convenience in Lemma~\ref{lemma:integrable}.)
The discrete Bessel point process has the following combinatorial interpretation.
Let~$\mathbb Y$ be the set of integer partitions (or, equivalently, Young diagrams).
Namely, elements $\lambda=(\lambda_1,\lambda_2,\dots)\in\mathbb Y$ are half-infinite sequences of non-negative integers $\lambda_i$, for $i\geq 1$, satisfying $\lambda_i\geq\lambda_{i+1}$ and with finitely many non-zero $\lambda_i$'s.
In particular, for $\lambda\in\mathbb Y$, the {\it weight} $|\lambda|:=\sum_{i\geq 1}\lambda_i$ is a finite number.
The {\it Poissonized Plancherel measure}~$\mathbb P_{\sf {Plan}}$ is the probability measure on~$\mathbb Y$, depending on a parameter~$L>0$, defined by
\begin{equation}
\label{eq:PoissPlancherel}
	\mathbb P_{\sf {Plan}}(\lbrace\lambda\rbrace) := {\rm e}^{-L^2}L^{2 | \lambda |}\left(\frac{\dim \lambda}{|\lambda|!}\right)^2 ,\qquad \lambda\in\mathbb Y.
\end{equation}
Here, $\dim \lambda$ is the dimension of the irreducibile representation of the symmetric group~$S_{|\lambda|}$ corresponding to $\lambda$, or, equivalently, $\dim \lambda$ is the number of standard Young tableaux of shape $\lambda$.
If we associate to each $\lambda\in\mathbb Y$ a subset of $\mathbb Z'$ through the map $\lambda \mapsto \lbrace\lambda_i - i + \tfrac 12 \rbrace_{i \geq 1}$, it was proven in~\cite{BOO,OkounkovWedge} that the push-forward of~$\mathbb P_{\sf {Plan}}$ is the determinantal point process on~$\mathbb Z'$ whose correlation kernel is precisely~\eqref{eq:kernelintegrable}.

\medskip

The kernel~\eqref{eq:K} has a similar interpretation when
\be
\label{eq:sigmaPP}
\sigma(l) = (1 + u^l)^{-1},\qquad l\in\Z',
\ee
for a parameter~$u \in [0,1)$.
Namely, introduce a probability measure~$\mathbb P_{\sf {cPlan}}$ on~$\mathbb Y$ ({\it cylindric Plancherel distribution}~\cite{Borodincylindric}), depending on parameters~$L>0$ and~$u\in[0,1)$, by
\be
\label{eq:cylindric}
\mathbb P_{\sf {cPlan}}(\lbrace\lambda\rbrace) := \frac 1{Z(u,L)}\sum_{\mu\subset\lambda} u^{|\mu|}
\biggl( \frac{\bigl(L(1-u)\bigr)^{|\lambda| - |\mu|}\dim(\lambda/\mu)}{(| \lambda | - | \mu |)!} \biggr)^2, \quad\lambda\in\mathbb Y,\ \quad
Z(u,L):=\frac{\e^{L^2(1-u)}}{\prod_{n\geq 1}(1-u^n)},
\ee
where the sum runs over partitions~$\mu\in\mathbb Y$ such that~$\mu_i\leq\lambda_i$ for all~$i\geq 1$, and~$\dim(\lambda/\mu)$ is the number of standard Young tableaux of shape~$\lambda/\mu$.
Consider also the probability measure~$\mathbb P_{\sf C}$ on~$\Z$ defined by
\begin{equation}
	\mathbb P_{\sf C}(\lbrace c\rbrace) = \frac{u^{c^2/2}}{\sum_{n \in \mathbb Z} u^{n^2/2}},\qquad c\in\Z.
\end{equation}
It is proven in~\cite{Borodincylindric,BeteaBouttier} that, under the map $(\lambda,C) \mapsto \{\lambda_i - i + 1/2 + C \}_{i \geq 1}$, the push-forward of~$\mathbb P_{\sf {cPlan}}\otimes\mathbb P_{\sf C}$ is the determinantal point process on~$\mathbb Z'$ whose correlation kernel is~\eqref{eq:K} with $\sigma$ as in~\eqref{eq:sigmaPP}.

\medskip

Going back to the kernel~\eqref{eq:K} for general~$\sigma$, we shall see in Lemma~\ref{lemma:defineDPP} that the induced determinantal point process has almost surely a largest particle~$a_{\sf max}$.
We shall study its cumulative distribution function
\be
Q_\s(L,s):=\mathbb P (a_{\sf max}\leq s),\qquad s\in\Z'.
\ee
By the general theory of determinantal point processes~\cite{SoshnikovDPP,JohanssonDPP,BorodinDPP}, this distribution can be expressed  as
\be
\label{eq:Fred}
Q_\s(L,s) = \det(1-\op P_s\op K_\s^\Be \op P_s),\qquad s\in\Z'.
\ee
Here, $\op K_\s^\Be$ is the operator on $\ell^2(\Z')$ induced\footnote{Throughout this paper, we agree that a kernel~$X:\Z'\times\Z'\to\C$ induces an operator $\op X$ on $\ell^2(\Z')$ by $(\op X \psi)(a)=\sum_{b\in \Z'}X(a,b)\psi(b)$, for $\psi\in\ell^2(\Z')$ and $a\in\Z'$.} by the kernel~\eqref{eq:K}, and $\op P_s$ is the orthogonal projector onto~$\ell^2(\{s+1,s+2\cdots\})$, namely, $\op P_s$ is induced by the kernel $P_s(a,b)=1_{a>s}\,\delta(a,b)$, for $s\in\Z'$.
The determinant in~\eqref{eq:Fred} is a Fredholm determinant, as the operator~$\op P_s\op K_\s^\Be \op P_s$ is trace-class on~$\ell^2(\Z')$ for all~$s\in\Z'$ (Lemma~\ref{lemma:defineDPP} below).

It is also worth noting that~$Q_\s(L,s)$ can be equivalently described as the following expectation with respect to the Poissonized Plancherel measure~\eqref{eq:PoissPlancherel} (see Lemma~\ref{lemma:123}):
\be
\label{eq:multstatPP}
Q_\s(L,s)=\mathbb E_{\sf {Plan}}\left[\prod_{i=1}^{+\infty}\bigl(1-\sigma(\lambda_i-i-s)\bigr)\right].
\ee

Finally, let us remark that~$0\leq Q_\s(L,s)\leq 1$ is a non-decreasing function of~$s\in\Z'$ such that $Q_\s(L,s)\to 1$ as~$s\to+\infty$.
In particular, there exists $s_0\in\Z'\cup\{-\infty\}$ (depending on~$\s$) such that $Q_\s(L,s)=0$ if~$s<s_0$ and $Q_\s(L,s)>0$ otherwise.
In particular, since for any~$\lambda\in\mathbb Y$ the set $\{\lambda_i-i+\tfrac 12\}_{i\geq 1}$ has largest particle $a_{\sf max}=\lambda_1-\tfrac 12\geq -\tfrac 12$, we deduce by the discussion above of the Poissonized Plancherel measure that, when~$\s=\mathbf 1_{\Z'_+}$, we have $s_0=-1/2$.
On the other hand, when~$\s(l)=(1+u^l)^{-1}$ as in~\eqref{eq:sigmaPP}, corresponding to the cylindric Plancherel measure, we have $s_0=-\infty$, because
\be
Q_\s(L,s)=\mathbb P(a_{\sf max}\leq s)\geq\mathbb P(a_{\sf max}=s)\geq \mathbb P_{\sf cPlan}\bigl(\lbrace\emptyset\rbrace\bigr)\mathbb P_{\sf C}(\lbrace s+\tfrac 12\rbrace)>0,\qquad\mbox{for all }s\in\Z'.
\ee

\medskip

Our first result is the following.

\begin{framed}
\begin{theoremintro}
\label{thm:1}
For all $s\in\Z'$ such that $Q_\s(L,s)>0$, we have
\be
\label{eq:thm1}
\frac{\pa^2}{\pa L^2}\log Q_\s(L,s)+\frac 1L\frac{\pa}{\pa L}\log Q_\s(L,s)+4=4\,\frac{Q_\s(L,s+1)Q_\s(L,s-1)}{Q_\s(L,s)^2}.
\ee
\end{theoremintro}
\end{framed}

The proof is given in Section~\ref{sec:thm1}.

The equation~\eqref{eq:thm1} is, essentially, a reduction of the {\it 2D Toda equation}.
Indeed, it implies that
\be
\tau_s(\theta_+,\theta_-):=\e^{\theta_+\theta_-}\,Q_\s(\sqrt{\theta_+\theta_-},s)
\ee
is a 2D Toda tau function, i.e., $\tau_s(\theta_+,\theta_-)$ satisfies the bilinear form of the 2D Toda equation~\cite{Hirota,UenoTakasaki}
\be
\frac{\pa^2}{\pa \theta_+\pa \theta_-}\log\tau_s(\theta_+,\theta_-) = \frac{\tau_{s+1}(\theta_+,\theta_-)\tau_{s-1}(\theta_+,\theta_-)}{\tau_s(\theta_+,\theta_-)^2}.
\ee
The equation~\eqref{eq:thm1}, or rather the corresponding equation for the variables $\{ {\rm e}^{L^2}Q_\s(L,s) \}_{s \in \mathbb Z'}$, is also known as \emph{cylindrical Toda equation}.
Another class of solutions of~\eqref{eq:thm1} written in terms of Fredholm determinants is studied in \cite{WcToda,TWcToda}.
More recently, using a Fredholm determinant representation, Matetski, Quastel, and Remenik proved that multi-point distributions associated to the polynuclear growth model with arbitrary initial data satisfy the {\it non-commutative Toda equation}~\cite{MQR}.

\medskip

It is appropriate to remark that, in the case~$\s=\mathbf 1_{X}$, with $X$ a subset of $\mathbb Z'$ bounded below\footnote{If $X$ is not bounded below, by~\eqref{eq:multstatPP} we have $Q_\s(L,s)=0$ identically.}, the connection to the 2D Toda equation is not new. Indeed, in this case, our result follows from~\cite[Theorem~3]{OkounkovWedge}, which relates more generally the Schur measure on partitions with the Toda hierarchy.
A particular case studied in even more detail is the one in which $X = \mathbb Z'_+$.
In this situation, by the combinatorial interpretation of the discrete Bessel point process explained above, we have $Q_\s(L,s)=0$ for $s\leq -\tfrac 32$ and $Q_\s(L,s)>0$ for $s\geq -\tfrac 12$.
Moreover, by the Borodin--Okounkov--Geronimo--Case formula~\cite{GC,BO}, the Fredholm determinant~$Q_\s(L,s)$, for~$s\in\Z'_+$, is related to a Toeplitz determinant of size~$\lceil s\rceil=s+\tfrac 12$ as
\be\label{eq:introToeplitz}
Q_\s(L,s)=\e^{-L^2}\det\left[{\rm I}_{i-j}(2L)\right]_{i,j=1,\dots,\lceil s\rceil}\,,
\qquad
{\rm I}_k(2L)=\res{u=0}\e^{L(u+u^{-1})}u^{-k-1}\d u.
\ee
Once this connection with Toeplitz determinants is established, the 2D Toda equation can be obtained in several different ways, essentially exploiting the relation with orthogonal polynomials on the unit circle, as for instance in \cite{AvMToda, Hisakado, BaikDeiftJohansson}.

Therefore, Theorem~\ref{thm:1} states that the connection of the discrete Bessel kernel to the 2D Toda equation extends to the deformation~\eqref{eq:K} of the kernel.
We complement this result by computing small~$L$ asymptotics for $Q_\s(L,s)$.

\begin{framed}
\begin{theoremintro}
\label{thm:ic}
For any $s\in\Z'$, let $Q^0_\sigma(s):=\prod_{i=1}^{+\infty}\bigl(1-\sigma(-i-s)\bigr)$.
For all $s\in\Z'$ such that $Q^0_\s(s)>0$, there exists $L_*=L_*(s)>0$ such that $Q_\s(L,s)>0$ for~$0\leq L<L_*$, and
\be
\label{eq:thmic}
\log Q_\s(L,s)=\log Q^0_\s(s)-\frac{\sigma(-s)-\sigma(-s-1)}{1-\sigma(-1-s)}L^2+\mathrm O(L^4),\qquad L\to 0.
\ee
\end{theoremintro}
\end{framed}

We note that when $Q^0_\s(s)>0$, the denominator in the term of order $L^2$ of~\eqref{eq:thmic} does not vanish. The proof is given in Section~\ref{sec:proofthmic}.

\subsection{Continuum limit to the Korteweg--de Vries equation}

The finite-temperature discrete Bessel kernels~\eqref{eq:K} have continuum limits to the finite-temperature Airy kernels~\cite{BeteaBouttier}.
These are kernels of the form
\be
\label{eq:airy}
K^{\sf Ai}_\varsigma(\xi,\eta;t)=\int_\R \varsigma(t^{-2/3}r)\Ai(\xi+r)\Ai(\eta+r)\d r,\qquad\xi,\eta\in\R,
\ee
with $\Ai$ and $\Ai'$ the Airy function and its derivative, respectively, $t>0$ a positive real parameter, and $\varsigma:\R\to[0,1]$ a function which is smooth and satisfies $\varsigma(r)\in L^1\bigl((-\infty,0),\sqrt{|r|}\d r\bigr)$.
In~\cite{BeteaBouttier}, the authors proved this limit for~$\s$ as in~\eqref{eq:sigmaPP}, but their result extends easily to more general functions, as long as $\sigma=\sigma_\epsilon$ depends on an additional parameter $\epsilon$ in such a way that $\sigma_\epsilon(\zeta/\epsilon)\to\varsigma(\zeta)$ for some function $\varsigma$ as $\epsilon\to 0$.
More precisely, when $\s$ is given by~\eqref{eq:sigmaPP}, one has to identify the parameter $\epsilon$ with $1-u$. Then, we have the convergence
\be
\sigma\biggl(\frac{r}{t^{2/3}(1-u)}\biggr)\to\varsigma\left(\frac{r}{t^{2/3}}\right)=\frac 1{1+\e^{-rt^{-2/3}}},\qquad\mbox{as }u\to 1^-,
\ee
which is the scaling limit used in~\cite{BeteaBouttier} to study the edge behavior of the cylindrical Plancherel measure.

\medskip

These type of kernels (and related Fredholm determinants) attracted a great deal of interest in the last 15 years. They first appeared in the field of random matrices~\cite{JohanssonTWGumb}, in the theory of the Kardar--Parisi--Zhang equation~\cite{AmirCorwinQuastel}, and in relation with one-dimensional systems of fermions at finite temperature~\cite{DeanLeDoussalMajumdarSchehr}.
Riemann--Hilbert techniques for the study of related Fredholm determinants have been developed and used in \cite{Bothner,CafassoClaeys,CafassoClaeysRuzza,BothnerCafassoTarricone, CharlierClaeysRuzza}.
In particular, Fredholm determinants on~$L^2(\R)$ of the form
\be
\label{eq:FredAiry}
F_\varsigma(x,t) = \det\bigl(1-\mathbf{1}_{(-xt^{-1/3},+\infty)}\op K^{\sf Ai}_\varsigma \mathbf{1}_{(-xt^{-1/3},+\infty)}\bigr)
\ee
have been shown to satisfy\footnote{Only equation~\eqref{eq:KdV} appears explicitly in \cite[Theorem~1.3]{CafassoClaeysRuzza}. However,~\eqref{eq:bilinearcKdV} can be obtained by substituting $U_\varsigma(x,t)=\partial_x^2\log F(x,t)+x/(2t)$ in~\eqref{eq:KdV} and integrating once in~$x$ thanks to the asymptotics proved in \cite[Section~5]{CafassoClaeysRuzza}.}~\cite{CafassoClaeysRuzza}
\be
\label{eq:bilinearcKdV}
\frac{\pa^2}{\pa t\pa x}\log F_\varsigma(x,t)+\frac xt\frac{\pa^2}{\pa x^2}\log F_\varsigma(x,t)+\left(\frac{\pa^2}{\pa x^2}\log F_\varsigma(x,t)\right)^2+\frac 16\frac{\pa^4}{\pa x^4}\log F_\varsigma(x,t)=0,
\ee
i.e., the function
\be
\label{eq:UKdV}
U_\varsigma(x,t):=\frac{\pa^2}{\pa x^2}\log F_\varsigma(x,t)+\frac x{2t}
\ee
satisfies the Korteweg--de\thinspace Vries equation
\be
\label{eq:KdV}
\frac{\pa}{\pa t}U_\varsigma(x,t)+2U_\varsigma(x,t)\frac{\pa}{\pa x}U_\varsigma(x,t)+\frac 16\frac{\pa^3}{\pa x^3}U_\varsigma(x,t)=0.
\ee

It is instructive to look at how equation~\eqref{eq:bilinearcKdV} (closely related to the bilinear form of the Korteweg--de\thinspace Vries equation) emerges in such continuum limit from equation~\eqref{eq:thm1} (which is in turn related to the 2D Toda equation).
Let the variables $L,s$ be given in terms of variables $x,t$ and of an additional parameter $\epsilon>0$ as
\be
\label{eq:continuumvariables}
s(x,t;\epsilon)=\frac 2{\epsilon^3t^2}-\frac{x}{\epsilon t},\quad L(x,t;\epsilon)=\frac 1{\epsilon^3t^2}.
\ee
Under this transformation, we have
\be
\frac{\pa}{\pa L}=\frac{\pa x}{\pa L}\frac{\pa}{\pa x}+\frac{\pa t}{\pa L}\frac{\pa}{\pa t}
=-\frac 12\epsilon t\left((\epsilon^2xt-4)\frac{\pa}{\pa x}+\epsilon^2t^2\frac{\pa}{\pa t}\right).
\ee
Moreover, let us introduce
\be
\label{eq:continuumF}
F(x,t;\epsilon) := Q_\s(L(x,t;\epsilon),s(x,t;\epsilon)).
\ee
As shown in~\cite{BeteaBouttier}, $F(x,t;\epsilon)$ converges, as~$\epsilon\to 0$, to $F_\varsigma(x,t)$, and we shall now explain how the equation for~$Q_\s$ of Theorem~\ref{thm:1} reduces to~\eqref{eq:bilinearcKdV}.
Expanding at $\epsilon=0$ as
\be
\label{eq:continuumLogF}
\log F(x,t;\epsilon)=f_0(x,t)+\epsilon f_1(x,t)+\epsilon^2 f_2(x,t)+\mathrm O(\epsilon^3),
\ee
the left-hand side of~\eqref{eq:thm1} is
\begin{align}
\nonumber
\biggl(\frac{\pa^2}{\pa L^2}+\frac 1L\frac\pa{\pa L}\biggr)\log F(x,t;\epsilon)+4={}&4+4\epsilon^2t^2\frac{\pa^2}{\pa x^2}f_0(x,t)
+4\epsilon^3t^2\frac{\pa^2}{\pa x^2}f_1(x,t)
\\
&
-2\epsilon^4t^4\biggr(\frac{\pa^2}{\pa t\pa x}f_0(x,t)+\frac xt\frac{\pa^2}{\pa x^2}f_0(x,t)-\frac 2{t^2}\frac{\pa^2}{\pa x^2}f_2(x,t)\biggl)+\mathrm O(\epsilon^5)
\end{align}
and, similarly, the right-hand side of~\eqref{eq:thm1} is
\begin{align}
\nonumber
4\frac{F(x-\epsilon t,t;\epsilon)F(x+\epsilon t,t;\epsilon)}{F(x,t;\epsilon)^2}={}&4+4\epsilon^2t^2\frac{\pa^2}{\pa x^2}f_0(x,t)+4\epsilon^3t^2\frac{\pa^2}{\pa x^2}f_1(x,t)
\\
&+\epsilon^4t^4\biggl(2\biggl(\frac{\pa^2}{\pa x^2}f_0(x,t)\biggr)^2+\frac 13\frac{\pa^4}{\pa x^4}f_0(x,t)+\frac 4{t^2}\frac{\pa^2}{\pa x^2}f_2(x,t)\biggr)+\mathrm O(\epsilon^5).
\end{align}
Terms of order up to~$\epsilon^3$ match identically, whilst at order~$\epsilon^4$ we obtain precisely~\eqref{eq:bilinearcKdV} (whose relation to the Korteweg--de\thinspace Vries equation has been explained above) for the function $F_\varsigma(x,t)=\exp\bigl(f_0(x,t)\bigr)$.

\begin{remark}
After submission, we learned that this scaling limit of the cylindrical Toda equation to the cylindrical KdV equation had already appeared in~\cite{Masuda}.
\end{remark}

\subsection{A discrete version of the integro-differential Painlev\'e II equation}

For the Korteweg--de\thinspace Vries solutions~$U_\varsigma(x,t)$ associated with Fredholm determinants~\eqref{eq:FredAiry} of the finite-temperature Airy kernel~\eqref{eq:airy} there is an identity between the potential and the wave-function\footnote{We thank Percy Deift for pointing out that such relation is the analogue of the {\it Trace Formula} of~\cite{DeiftTrubowitz} for potentials $U_\varsigma(x,t)$ which, unlike the classical setting of op. cit., do not vanish as $x\to\pm\infty$ but rather behave as $x/(2t)$.}.
Namely, provided exponential decay of~$\varsigma$ at~$-\infty$, it is shown in~\cite{CafassoClaeysRuzza} that the solution to the boundary value problem
\be
\label{eq:continuumSchroedinger}
\frac{\pa^2}{\pa x^2}\psi(\zeta;x,t) = (\zeta-2U_\varsigma(x,t))\psi(\zeta;x,t),\qquad\psi(\zeta;x,t)\sim t^{1/6}\Ai(t^{2/3}\zeta-xt^{-1/3}),\ \ x\to-\infty,
\ee
satisfies
\be
\label{eq:traceidentity}
U_\varsigma(x,t)=\frac x{2t}-\frac 1t\int_\R\psi(\eta;x,t)^2\varsigma'(\eta)\d\eta.
\ee
Plugging~\eqref{eq:traceidentity} into~\eqref{eq:continuumSchroedinger} one obtains the so-called {\it integro-differential Painlev\'e II equation} of Amir, Corwin, and Quastel~\cite{AmirCorwinQuastel}
\be
\label{eq:ACQ}
\frac{\pa^2}{\pa x^2}\psi(\zeta;x,t) = \left(\zeta-\frac xt+\frac 2t\int_\R\psi(\eta;x,t)^2\varsigma'(\eta)\d\eta\right)\psi(\zeta;x,t),
\ee
whose solution (subject to the boundary value condition in~\eqref{eq:continuumSchroedinger}) characterizes the distribution $F_\varsigma$, since, by~\eqref{eq:UKdV} and~\eqref{eq:traceidentity},
\be
\label{eq:airyfrednonlocal}
\frac{\pa^2}{\pa x^2}\log F_\varsigma(x,t)=-\frac 1t\int_\R\psi(\eta;x,t)^2\varsigma'(\eta)\d\eta.
\ee

It is worth recalling that in the limit~$\varsigma\to\mathbf 1_{(0,+\infty)}$, the kernel~\eqref{eq:airy} reduces to the classical Airy kernel, the integro-differential Painlev\'e II equation~\eqref{eq:ACQ} reduces to the standard Painlev\'e II equation, and its solution selected by the boundary behavior in~\eqref{eq:continuumSchroedinger} is the Hastings--McLeod solution (in agreement with the celebrated result by Tracy and Widom~\cite{TracyWidom}). 

The next result is an analogous property for the finite-temperature discrete Bessel kernels.

\begin{framed}
\begin{theoremintro}
\label{thm:2}
Let $L>0$ and $s_0:=\min\{s\in\Z':\ Q_\s(L,s)>0\}\in\Z'\cup\{-\infty\}$.
For all $s\in\Z'$ with $s\geq s_0$, we introduce
\be
\label{eq:frakthm}
\mathfrak a(L,s):=\frac{\sqrt{Q_\s(L,s+1)Q_\s(L,s-1)}}{Q_\s(L,s)},\qquad
\mathfrak b(L,s+1):=\frac{\pa}{\pa L}\log\frac{Q_\s(L,s+1)}{Q_\s(L,s)}.
\ee
Then, for all $s\in\Z'$, $s\geq s_0$,
\begin{align}
\label{eq:frakanonlocal}
\mathfrak a^{-1}(L,s)-\mathfrak a(L,s)&=\frac 1L\sum_{l\in\Z'}\bigl(\s(l+1)-\s(l)\bigr)\varphi(l+1;L,s-1)\varphi(l;L,s),
\\
\label{eq:frakbnonlocal}
\mathfrak b(L,s+1)&=\frac 2L\sum_{l\in\Z'}\bigl(\s(l+1)-\s(l)\bigr)\varphi(l+1;L,s)\varphi(l;L,s),
\end{align}
where $\varphi(l;L,s)$ are defined for $l\in\Z'$ and for $s\in\Z'$ with $s\geq s_0-1$ and satisfy the recursion
\be
\label{eq:eqvarphithm}
\mathfrak a(L,s+1)\varphi(l;L,s+1)+\mathfrak a(L,s)\varphi(l;L,s-1)=\left(\frac{l+s+1}L+\frac{\mathfrak b(L,s+1)}2\right)\varphi(l;L,s).
\ee
Moreover, for all $l\in\Z'$ we have
\be
\label{eq:asympvarphi}
\varphi(l;L,s)\sim \sqrt L\,\J_{l+s+1}(2L),\qquad s\to+\infty.
\ee
\end{theoremintro}
\end{framed}

The proof is given in Section~\ref{sec:thm2} and is based on a Lax pair argument. In particular, when $\sigma=\mathbf 1_{\Z'_+}$ we obtain a Lax pair which, although different from the one used by Borodin~\cite{Borodin}, can be equivalently used to prove the connection to the discrete Painlev\'e II equation established in op.~cit. (and independently proved by other methods in~\cite{Baik,AvMToda}).
See Section~\ref{sec:dPII} for more details.

It is worth observing that in the scaling limit~\eqref{eq:continuumvariables} as~$\epsilon\to 0$, the equations of Theorem~\ref{thm:2} formally reduce to above mentioned equations for Fredholm determinants of the finite-temperature Airy kernel.
More precisely, with the notations of~\eqref{eq:continuumvariables},~\eqref{eq:continuumF}, and~\eqref{eq:continuumLogF}, we have, as~$\epsilon\to 0$,
\begin{align}
\label{eq:atoone}
\mathfrak a\bigl(L(x,t;\epsilon),s(x,t;\epsilon)\bigr)&= 1+\frac12\epsilon^2t^2\frac{\pa^2}{\pa x^2}f_0(x,t)+\mathrm O(\epsilon^3),
\\
\label{eq:limitlhs1}
\mathfrak a^{-1}\bigl(L(x,t;\epsilon),s(x,t;\epsilon)\bigr)-\mathfrak a\bigl(L(x,t;\epsilon),s(x,t;\epsilon)\bigr)&\sim-\epsilon^2t^2\frac{\pa^2}{\pa x^2}f_0(x,t),
\\
\label{eq:limitlhs2}
\mathfrak b\bigl(L(x,t;\epsilon),s(x,t;\epsilon)+1\bigr)&\sim-2\epsilon^2t^2\frac{\pa^2}{\pa x^2}f_0(x,t).
\end{align}
Introducing $\psi$ and $\varsigma$ by the $\epsilon\to 0$ expansions
\be
(\epsilon t)^{1/2}\varphi(\zeta/\epsilon;L(x,t;\epsilon),s(x,t;\epsilon))=\psi(\zeta;x,t)+\mathrm O(\epsilon),\qquad \sigma(\zeta/\epsilon)=\varsigma(\zeta)+\mathrm O(\epsilon),
\ee
we also have (by approximating a Riemann--Stieltjes sum with the corresponding integral)
\begin{align}
\label{eq:limitrhs1}
\frac 1L\sum_{l\in\Z'}\bigl(\s(l+1)-\s(l)\bigr)\varphi(l+1;L,s-1)\varphi(l;L,s)\Big|_{L=L(x,t;\epsilon),\ s=s(x,t;\epsilon)}&\sim\epsilon^2t\int_\R\varsigma'(\eta)\psi(\eta;x,t)^2\d\eta,
\\
\label{eq:limitrhs2}
\frac 2L\sum_{l\in\Z'}\bigl(\s(l+1)-\s(l)\bigr)\varphi(l+1;L,s)\varphi(l;L,s)\Big|_{L=L(x,t;\epsilon),\ s=s(x,t;\epsilon)}&\sim 2\epsilon^2t\int_\R\varsigma'(\eta)\psi(\eta;x,t)^2\d\eta,
\end{align}
By~\eqref{eq:frakanonlocal} we have equality of~\eqref{eq:limitlhs1} and ~\eqref{eq:limitrhs1}, and looking at the leading order terms gives~\eqref{eq:airyfrednonlocal}.
Similarly, by~\eqref{eq:frakbnonlocal} we have equality of~\eqref{eq:limitlhs2} and~\eqref{eq:limitrhs2}, and looking at the leading order terms gives again~\eqref{eq:airyfrednonlocal}.
Moreover, using~\eqref{eq:atoone} and~\eqref{eq:limitlhs2}, equation~\eqref{eq:eqvarphithm} reduces to~\eqref{eq:continuumSchroedinger}.
Finally, also the asymptotic relation in~\eqref{eq:asympvarphi} for $\varphi$ formally matches with the one for $\psi$ in~\eqref{eq:continuumSchroedinger} by using~\cite[Lemma~4.4]{BOO}
\be
L^{1/3}\J_{2L+\xi L^{1/3}}(2L)\sim \Ai(\xi),\qquad L\to+\infty.
\ee

\subsection{Organization of the rest of the paper}
In Section~\ref{sec:2} we gather some properties of the discrete Bessel point process and its finite-temperature deformation.
In Section~\ref{sec:RH} we prove a {\it discrete Riemann--Hilbert} characterization of~$Q_\s(L,s)$, following a general strategy developed by Borodin and Deift~\cite{BorodinIMRN} which parallels the theory of integrable operators of Its--Izergin--Korepin--Slavnov~\cite{IIKS} in a discrete setting.
Next, we prove Theorem~\ref{thm:1}, \ref{thm:ic}, and~\ref{thm:2} in Sections~\ref{sec:thm1}, \ref{sec:proofthmic}, and~\ref{sec:thm2}, respectively.
We briefly discuss the connection of our approach to the results of Borodin~\cite{Borodin} relative to the special case $\s=\mathbf 1_{\Z_+'}$ in Section~\ref{sec:dPII}.
An elementary technical lemma which is helpful in the discussion of discrete Riemann--Hilbert problems is deferred to Appendix~\ref{app:partialfraction}.

\section{Preliminaries on the discrete Bessel kernel}\label{sec:2}

The Bessel functions satisfy~\cite[eq.~10.6.1]{DLMF}
\be
\label{eq:diffeqBessel}
L\bigl(\J_{k+1}(2L)+\J_{k-1}(2L)\bigr)=k\J_k(2L),\qquad k\in\C,
\ee
and~\cite[eq.~10.4.1]{DLMF}
\be
\label{eq:symmetryBessel}
\J_{-k}(2L)=(-1)^k\J_k(2L),\qquad k\in\Z.
\ee

\begin{lemma}
As $k\to+\infty$, we have
\begin{align}
\label{eq:1lineasymptotics}
\J_k(2L)&\sim\frac{1}{\sqrt{2\pi k}}\left(\frac{\e L}k\right)^k,&
\frac{\pa}{\pa \kappa}\J_\kappa(2L)\biggr|_{\kappa=k}&\sim \frac{\log(L/k)}{\sqrt{2\pi k}}\left(\frac{\e L}k\right)^k,
\intertext{and, as $k\to+\infty$ through integer values, we have}
\label{eq:2lineasymptotics}
\J_{-k}(2L)&\sim\frac{(-1)^k}{\sqrt{2\pi k}}\left(\frac{\e L}k\right)^k,&
\frac{\pa}{\pa \kappa}\J_\kappa(2L)\biggr|_{\kappa=-k}&\sim (-1)^{k+1}\sqrt{\frac{2\pi}{k}}\left(\frac k{\e L}\right)^k.
\end{align}
\end{lemma}
\begin{proof}
For real $k>-1/2$, we can represent the Bessel function by the Poisson integral~\cite[eq.~10.9.4]{DLMF}
\be
\label{eq:Poisson}
\J_k(2L)=\frac{L^k}{\sqrt\pi\,\Gamma(k+\tfrac 12)}\int_0^\pi\cos(2L\cos\theta)\e^{2k\log(\sin\theta)}\d\theta.
\ee
Since~$\theta\mapsto\log(\sin\theta)$ has a unique non-degenerate maximum at~$\theta=\pi/2$ for $\theta\in(0,\pi)$, it suffices to use Laplace's method to obtain the large~$k$ asymptotics of the integral. Combining with Stirling's asymptotics, we obtain the first relation in~\eqref{eq:1lineasymptotics}.
The first relation in~\eqref{eq:2lineasymptotics} then follows from~\eqref{eq:symmetryBessel}.

Next, by~\eqref{eq:Poisson}, for real $k>-1/2$,
\be
\frac\pa{\pa k}\J_k(2L)=\frac\pa{\pa k}\biggl(\log\frac{L^k}{\sqrt\pi\,\Gamma(k+\tfrac 12)}\biggr)\J_k(2L)+\frac{L^k}{\sqrt\pi\,\Gamma(k+\tfrac 12)}\int_0^\pi 2\log(\sin\theta)\cos(2L\cos\theta)\e^{2k\log(\sin\theta)}\d\theta.
\ee
Using the asymptotics for the digamma function $\G'/\G$, as well as the already established first relation in~\eqref{eq:1lineasymptotics} for the first term, and again Laplace's method for the second term, we obtain (after some computations) the second relation in~\eqref{eq:1lineasymptotics}.
Finally, for the last relation we use that, for $k\in\Z$, we have~\cite[10.2.4]{DLMF},
\be
(-1)^k\frac{\pa}{\pa \kappa}\J_\kappa(2L)\biggr|_{\kappa=-k}=\pi\mathrm Y_k(2L)-\frac{\pa}{\pa \kappa}\J_\kappa(2L)\biggr|_{\kappa=k}
\ee
where $\mathrm Y_k(\cdot)$ is the Bessel function of second kind of order $k$, and it suffices to use the second relation in~\eqref{eq:1lineasymptotics} along with the asymptotics~$\mathrm Y_k(2L)\sim -\sqrt{2/(\pi k)}\,(\e L/k)^{-k}$ as~$k\to+\infty$~\cite[eq.~10.19.2]{DLMF}.
\end{proof}

Let us recall the discrete Bessel kernel~$K^\Be(a,b)=\sum_{l\in\Z'_+}\J_{a+l}(2L)\J_{b+l}(2L)$, as in~\eqref{eq:kernelintegrable}.
It is worth observing that only Bessel functions of integer order appear in this expression. 

\begin{lemma}
\label{lemma:integrable}
We have
\begin{align}
\label{eq:offdiagonal}
K^\Be(a,b)&=L\,\frac{\J_{a-\frac 12}(2L)\J_{b+\frac 12}(2L)-\J_{a+\frac 12}(2L)\J_{b-\frac 12}(2L)}{a-b}\,,&
&a,b\in\Z',\ a\not=b,
\\
\label{eq:diagonal}
K^\Be(a,a)&=L\biggl(\J_{a+\frac 12}(2L)\frac{\pa\J_{a-\frac 12}(2L)}{\pa a}-\J_{a-\frac 12}(2L)\frac{\pa\J_{a+\frac 12}(2L)}{\pa a}\biggr)\,,& &a\in\Z'.
\end{align}
In particular, $K^\Be(a,a)$ is a decreasing function of $a\in\Z'$ satisfying
\be
\label{eq:diagtoone}
K^\Be(a,a)\to 1,\qquad a\to-\infty,\ a\in\Z'.
\ee
\end{lemma}
\begin{proof}
Fix $M\in\Z'_+$.
Using~\eqref{eq:diffeqBessel} we compute, for any real $a\not=b$, omitting the argument~$2L$ of the Bessel functions,
\begin{align}
\nonumber
(a-b)\sum_{l \in \mathbb Z'_+\cap[\tfrac 12,M]} \J_{a + l} \J_{b + l}&=\sum_{l \in \mathbb Z'_+\cap[\tfrac 12,M]} \left((a+l) \J_{a + l} \J_{b + l}-(b+l)\J_{a + l} \J_{b + l}\right)
\\
\nonumber
&=L\sum_{l \in \mathbb Z'_+\cap[\tfrac 12,M]} \left(\J_{a + l -1} \J_{b + l}+\J_{a + l +1} \J_{b + l}-\J_{a + l} \J_{b + l -1}-\J_{a + l}\J_{b + l +1}\right)
\\
\label{eq:computationkernelintegrable}
&=L\bigl(\J_{a-\frac 12}\J_{b+\frac 12}+
\J_{a+M+1} \J_{b + M}-\J_{a+\frac 12}\J_{b-\frac 12}-\J_{a + M}\J_{b + M+1}\bigr),
\end{align}
where in the last step we telescope the sum.
Sending $M\to+\infty$ and using the first asymptotics in~\eqref{eq:1lineasymptotics} and~\eqref{eq:2lineasymptotics}, we obtain~\eqref{eq:offdiagonal}.
Sending instead $a\to b$ first and then sending $M\to+\infty$ we obtain~\eqref{eq:diagonal}.
Finally, it suffices to insert~\eqref{eq:2lineasymptotics} in~\eqref{eq:diagonal} to obtain~\eqref{eq:diagtoone}.
\end{proof}

\begin{lemma}
\label{lemma:orthoBessel}
For all~$a,b\in\Z'$ we have $\sum_{l\in\Z'}\J_{a+l}(2L)\J_{b+l}(2L) = \delta_{a,b}$.
\end{lemma}
\begin{proof}
Let $M,N\in\Z'$ with $N<0<M$. By using a similar argument as in~\eqref{eq:computationkernelintegrable}, we obtain, for real $a\not=b$,
\be
(a-b)\sum_{l \in\Z'\cap[N,M]}\J_{a+l}\J_{b+l}=L\bigl(\J_{a+N-1}\J_{b+N}-\J_{a+N}\J_{b+N-1}+
\J_{a+M+1} \J_{b + M}-\J_{a + M}\J_{b + M+1}\bigr),
\ee
and so sending $M\to+\infty,N\to-\infty$ and using~\eqref{eq:1lineasymptotics} we obtain the thesis for~$a\not=b$. Sending instead $a\to b$ first, and then sending $M\to+\infty,N\to-\infty$ using~\eqref{eq:2lineasymptotics} and~\eqref{eq:diagtoone} we obtain the thesis for~$a=b$.
\end{proof}

\begin{lemma}
\label{lemma:defineDPP}
We have $0\leq\op K_\s^\Be\leq 1$.
Moreover, if~$\s\in\ell^1(\Z'\cap(-\infty,0))$, the operator $\op P_s\mathcal K_\s^\Be\op P_s$ is trace-class, where $\op P_s$ is the orthogonal projector onto~$\ell^2(\lbrace s+1,s+2,\dots\rbrace)$, for all $s\in\Z'$.
\end{lemma}
\begin{proof}
It follows from Lemma~\ref{lemma:orthoBessel} that the operator $\op J$ induced by the kernel $\J_{a+b}(2L)$ is an unitary involution of~$\ell^2(\Z')$, i.e.~$\op J=\op J^\dagger=\op J^{-1}$. By a slight abuse of notation, denote with $\sigma$ the operator of multiplication by $\s$, i.e. the operator on~$\ell^2(\Z')$ induced by the kernel $\sigma(a)\delta(a,b)$.
Then, by definition,~$\op K_\s^\Be=\op J\s\op J$.
Let $\langle\cdot,\cdot\rangle$ be the scalar product on~$\ell^2(\Z')$: since $0\leq \sigma\leq 1$, we have
\be
\langle(\op K_\s^\Be)^2\psi,\psi\rangle=\langle\s^2\op J\psi,\op J\psi\rangle\leq \langle\s\op J\psi,\op J\psi\rangle = \langle\op K_\s^\Be\psi,\psi\rangle,\qquad\mbox{for all }\psi\in\ell^2(\Z').
\ee
Therefore, $\op K_\s^\Be\geq (\op K_\s^\Be)^2\geq 0$, which also implies~$1-\op K_\s^\Be\geq(1-\op K_\s^\Be)^2\geq 0$.

For the second statement, observe that~$\op P_s\op K_\s^\Be\op P_s=\op H_s\op H^\dagger_s$ where~$\op H_s$ is induced by kernel 
\be
\label{eq:H}
H_s(a,b)=\mathbf 1_{a>s}\J_{a+b}(2L)\sqrt{\s}(b),\qquad a,b\in\Z'.
\ee
For a fixed~$s\in\Z'$, the operator~$\op H_s$ is Hilbert--Schmidt on~$\ell^2(\Z')$ if and only if
\be
\label{eq:proofcondition}
\sum_{a,b\in\Z'}|H_s(a,b)|^2
=\sum_{a\in\Z'_+}\sum_{b\in\Z'}\mathbf 1_{a>s}\J_{a+b}(2L)^2\s(b)
=L\sum_{l\in\Z'}\s(l-s-\tfrac 12)K^\Be(l,l)<+\infty.
\ee
The convergence of the latter series at~$l \to+\infty$ follows from~\eqref{eq:diagonal} and the first asymptotic relations in~\eqref{eq:1lineasymptotics} and~\eqref{eq:2lineasymptotics}, along with~$0\leq\sigma\leq 1$.
The convergence at~$l \to -\infty$ follows instead by~\eqref{eq:diagonal} and the second asymptotic relations in~\eqref{eq:1lineasymptotics} and~\eqref{eq:2lineasymptotics}, along with the summability assumption on~$\sigma$.
\end{proof}

It follows from this lemma and the Macchi--Soshnikov criterion~\cite[Theorem~3]{SoshnikovDPP} that there exists a unique determinantal point process on~$\Z'$ whose correlation kernel is~$K_\s^\Be$.
It also follows from the general theory of determinantal point processes, e.g., from~\cite[Theorem~4]{SoshnikovDPP}, that this process has almost surely a largest particle $a_{\sf max}$, whose distribution is given by the Fredholm determinant as in~\eqref{eq:Fred}.

\section{Discrete Riemann--Hilbert characterization of \texorpdfstring{$Q_\s$}{Q}}\label{sec:RH}

Let us introduce the operator $\op M_s$ on $\ell^2(\Z')$, for $s\in\Z'$, induced by the kernel
\be
\label{eq:Lkernel}
M_s(a,b)=\sqrt{\s}\bigl(a-s-\tfrac 12\bigr)K^\Be(a,b)\sqrt{\s}\bigl(b-s-\tfrac 12\bigr),\qquad a,b\in\Z'.
\ee
where $K^\Be$ is as in~\eqref{eq:kernelintegrable}.
The operator~$\op M_s$ is of {\it discrete integrable form}~\cite{BorodinIMRN,Borodin}, namely, using~\eqref{eq:offdiagonal} the off-diagonal entries of the kernel can be expressed as
\be
\label{eq:Mintegrable}
M_s(a,b)=\frac{\mathbf f^\top(a)\mathbf g(b)}{a-b},\qquad a,b,\in\Z',\ a\not=b,
\ee
where
\be
\label{eq:fg}
\mathbf f(a):=\sqrt{\s}(a-s-\tfrac 12)\begin{pmatrix}
\J_{a-\frac 12}(2L) \\ L\J_{a+\frac 12}(2L)
\end{pmatrix},\qquad
\mathbf g(b):=\sqrt{\s}(b-s-\tfrac 12)\begin{pmatrix}
L \J_{b+\frac 12}(2L) \\ -\J_{b-\frac 12}(2L)\end{pmatrix}\,.
\ee
Using~\eqref{eq:diagonal} we can express the diagonal entries as
\be
\label{eq:Mdiagonal}
M_s(a,a)=\s(a-s-\tfrac 12) K^\Be(a,a)=
L\,\s(a-s-\tfrac 12)\,\biggl(\J_{a+\frac 12}(2L)\frac{\pa \J_{a-\frac 12}(2L)}{\pa a}-\J_{a-\frac 12}(2L)\frac{\pa \J_{a+\frac 12}(2L)}{\pa a}\biggr).
\ee

\begin{lemma}
\label{lemma:123}
\noindent{{\it(i)}}
The operator $\op M_s$ is trace-class and we have
\be
\label{eq:lemma1}
Q_\s(L,s)=\det(1-\op M_s).
\ee

\noindent{{\it(ii)}}
The identity~\eqref{eq:multstatPP} holds true.

\noindent{{\it(iii)}}
For all $s\in\Z'$ such that $Q_\s(L,s)>0$, we have
\be
\label{eq:lemma2}
\frac{Q_\s(L,s-1)}{Q_\s(L,s)}-1=\tr\left((1-\op M_s)^{-1}\op N_s\right),
\ee
where $\op N_s$ is the rank one operator on $\ell^2(\Z')$ induced by the kernel
\be
\label{eq:rankone}
N_s(a,b)=\sqrt{\s}(a-s-\tfrac 12)\,\J_{a-\frac 12}(2L)\,\sqrt{\s}(b-s-\tfrac 12)\,\J_{b-\frac 12}(2L).
\ee
\end{lemma}
\begin{proof}
\noindent{{\it(i)}}
We have
\be
\label{eq:computationlemma}
Q_\sigma(L,s)=\det(1-\op P_s\op K^\Be_\s \op P_s)=\det(1-\op H_s\op H_s^\dagger)=\det(1-\op H_s^\dagger\op H_s),
\ee
where $\op H_s$ is induced by the kernel~\eqref{eq:H}.
Let $\op T$ be the shift operator on $\ell^2(\Z')$, induced by the kernel $T(a,b)=\delta_{a,b+1}$.
It is straightforward to verify that $\op H_s^\dagger\op H_s=\op T^{s+\tfrac 12}\op M_s\op T^{-s-\tfrac 12}$.
This identity implies that $\op M_s$ is trace-class, and, by combining it with \eqref{eq:computationlemma}, we obtain~\eqref{eq:lemma1}.

\noindent{{\it(ii)}}
We have, by the previous point,
\be
Q_\s(L,s) = \det(1-\sigma(\cdot-s-\tfrac 12)\op K^\Be)
\ee
where~$\sigma(\cdot-s-\tfrac 12)$ denotes the multiplication operator induced by the kernel~$\sigma(a-s-\tfrac 12)\delta_{a,b}$.
Then~\eqref{eq:multstatPP} follows from a general property of determinantal point processes (e.g., see~\cite[eq.~(11.2.4)]{BorodinDPP}).

\noindent{{\it(iii)}}
Let~$\op S_s$ be the operator on~$\ell^2(\Z')$ induced by the kernel~$S_s(a,b)=\sqrt\s(a-s-\tfrac 12)\delta_{a,b}$.
Then, $\op M_s=\op S_s\op K^\Be\op S_s$ where $\op K^\Be$ is the operator induced by the discrete Bessel kernel $K^\Be$, defined in~\eqref{eq:kernelintegrable}.
Recalling the shift operator~$\op T$, induced by the kernel $T(a,b)=\delta_{a,b+1}$, we observe that $\op S_{s-1}=\op T^{-1}\op S_s\op T$ so that
\be
Q_\s(L,s-1)=\det\left(1-\op S_s\op T\op K^\Be\op T^{-1}\op S_s\right)=
\det\left(1-\op M_s+\op S_s\left(\op K^\Be-\op T\op K^\Be\op T^{-1}\right)\op S_s\right).
\ee
From~\eqref{eq:kernelintegrable}, we note that $\op N_s:=\op S_s\left(\op K^\Be-\op T\op K^\Be\op T^{-1}\right)\op S_s$ is the rank one operator induced by the kernel~\eqref{eq:rankone}.
As long as $Q_\s(L,s)\not=0$, we have
\begin{align}
\nonumber
Q_\s(L,s-1)=\det(1-\op M_s+\op N_s)&=\det(1-\op M_s)
\,\det\bigl(1+(1-\op M_s)^{-1}\op N_s\bigr)
\\
&=Q_\s(L,s)\biggl(1+\tr\bigl((1-\op M_s)^{-1}\op N_s\bigr)\biggr),
\end{align}
by using a standard formula for the determinant of a rank-one perturbation of the identity.
\end{proof}

The next key step is to apply the discrete version of Its--Izergin--Korepin--Slavnov procedure \cite{IIKS}, as developed for instance by Borodin~\cite{Borodin}.
This approach provides us with an effective way of computing the resolvent operator $\op R_s:=(1-\op M_s)^{-1}-1$ that proves useful to investigate~\eqref{eq:lemma2}.
Indeed, the main result of this theory (Theorem~\ref{thm:Borodin} below, following from general results of Borodin) is that the resolvent operator $\op R_s$ is also induced by a kernel of integrable form expressed through a meromorphic $2\times 2$ matrix-valued function $Y(\cdot)$ (parametrically depending on $\s,s,L$ as well) which is uniquely characterized by the following \emph{discrete Riemann--Hilbert} (RH) conditions.

\subsubsection*{Discrete RH problem for \texorpdfstring{$Y$}{Y}}
\begin{itemize}
\item[(a)] {\em $Y(z)$ is a $2\times 2$ matrix-valued meromorphic function of $z$ with simple poles at $\Z'$ only.}
\item[(b)] {\em For all $a\in\Z'$, the function
\be 
Y_a^\reg(z):=Y(z)\left(I-\frac{W_Y(a)}{z-a}\right)
\ee
has a removable singularity at $z=a$, where
\be
\label{eq:WY}
W_Y(a):=\frac{\mathbf f(a)\mathbf g^\top(a)}{1-M_{s}(a,a)},\qquad a\in\Z'.
\ee
Here, $\mathbf f(a)$, $\mathbf g(a)$, and $M_s(a,a)$ are given explicitly in~\eqref{eq:fg} and~\eqref{eq:Mdiagonal}.
}
\item[(c)] {\em We have $\lim_{n\to+\infty}\sup_{|z|=n}|Y(z)-I|=0$, where the limit is taken over integer values of $n$, $I$ denotes the $2\times 2$ identity matrix and $|\cdot|$ denotes any matrix-norm.}
\end{itemize}

Before describing how $Y$ allows us to express the resolvent operator $\op R_s$, we make a few observations.

\begin{remark}
\label{remark:dRHp}
\begin{itemize}
\item[{\it (i)}] The usual formulation of condition (b) in the discrete RH problem is the slightly different but completely equivalent requirement that, for all $a\in\Z'$, the limit $\lim_{z\to a}Y(z)W_Y(a)$ exists and that
\be
\label{eq:bstandard}
\lim_{z\to a}Y(z)W_Y(a) = \res{z=a}Y(z)\,\d z.
\ee

\item[{\it (ii)}] Since~$0\leq \s(a)\leq 1$ and~$K^\Be(a,a)<1$ for all~$a\in\Z'$ (see Lemma~\ref{lemma:integrable}), we get $1-M_s(a,a)>0$ for all $a\in\Z'$.
In particular, \eqref{eq:WY} is well defined.

\item[{\it (iii)}] For any solution $Y$ to the above discrete RH problem, we have $\det Y(z)=1$ identically in $z$. Indeed, $\mathbf f^\top(a)\mathbf g(a)=0$ implies $W_Y^2(a)=0$, hence $\det Y(z)=\det Y_a^\reg(z)$ for all $a\in\Z'$ and so $\det Y(z)$ extends to an entire function of $z$.
By condition~(c) together with the maximum modulus theorem we conclude that $\det Y(z)=1$ identically in $z$.

\item[{\it (iv)}] The solution $Y$ to the above discrete RH problem is unique, if any exists.
Indeed, for any two solutions $Y(z)$ and $\wt Y(z)$, the matrix $T(z):=\wt Y(z)Y^{-1}(z)$ has removable singularities at $\Z'$ by condition~(b), because $T(z)=\wt Y_a^\reg(z)(Y_a^\reg)^{-1}(z)$ for all $a\in\Z'$, hence $T(z)$ extends to an entire matrix function of $z$.
By condition~(c) together with the maximum modulus theorem, we infer that $T(z)=I$ identically in $z$.

\item[{\it (v)}] The condition~(b) in the discrete RH problem for $Y$ implies that $Y(z)$ has the following Laurent expansion near $z=a\in\Z'$:
\be
\label{eq:LaurentY}
Y(z)=C_Y(a)\left(\frac{W_Y(a)}{z-a}+I+Y_1(a)(z-a)+O\bigl((z-a)^2\bigr)\right),
\ee
where $C_Y(a)$ is an invertible matrix.
In particular, although $Y(z)$ has a pole as $z\to a\in\Z'$, the limits $\lim_{z\to a}Y(z)\mathbf f(a)$ and $\lim_{z\to a}Y^{-\top}(z)\mathbf g(a)$ for $a\in\Z'$ exist and are finite.
In the interest of lighter notations, we suppress the limit notation in such expressions, namely for $a\in\Z'$ we define
\be
Y(a)\mathbf f(a):=\lim_{z\to a}Y(z)\mathbf f(a),\quad Y^{-\top}(a)\mathbf g(a):=\lim_{z\to a}Y^{-\top}(z)\mathbf g(a).
\ee
Similarly, for $a\in\Z'$ we also define
\be
Y'(a)\mathbf f(a):=\lim_{z\to a}\frac{\d Y(z)}{\d z}\mathbf f(a)=C_Y(a)Y_1(a)\mathbf f(a).
\ee
Similarly, the inverse matrix $Y^{-1}$ has the Laurent expansion
\be
\label{eq:LaurentYinverse}
Y^{-1}(z)=\left(-\frac{W_Y(a)}{z-a}+I+\wt Y_1(a)(z-a)+O\bigl((z-a)^2\bigr)\right)\wt C_Y(a),
\ee
where $\wt C_Y(a)$ is an invertible matrix, which does not necessarily coincide with $C_Y^{-1}(a)$.

\item[{\it (vi)}]
In what follows we shall need also the subleading terms in the expansion at $z\to\infty$:
\be
\label{eq:Yinfty}
Y(z)=I+\begin{pmatrix}
\a & \b \\ \g & -\a
\end{pmatrix}z^{-1}+O(z^{-2}),
\ee
for functions $\a=\a(L,s)$, $\b=\b(L,s)$ and $\g=\g(L,s)$. This matrix is traceless because $\det Y(z)=1$ identically in $z$.
Here, as in condition~(c) of the discrete RH problem, $|z|\to+\infty$ through integer values.
\end{itemize}
\end{remark}

\begin{lemma}
\label{lemma:structure}
Fix $a\in\Z'$.
Let $C_Y(a)$ and $Y_1(a)$ be as in~\eqref{eq:LaurentY}, and let $c_Y(a):=\det C_Y(a)$.
We have
\be
\label{eq:structuredet}
\frac{{\mathbf g}^\top(a)Y_1(a){\mathbf f}(a)}{1-M_{s}(a,a)} = \frac{c_Y(a)-1}{c_Y(a)}
\ee
and, for some $d_Y(a)\in\C$,
\be
\label{eq:structureCC}
\wt C_Y(a) C_Y(a)=c_Y(a)I+d_Y(a)W_Y(a).
\ee
\end{lemma}
\begin{proof}
Since $\mathbf f(a),\mathbf g(a)$ are orthogonal and nonzero, the $2\times 2$ matrix
\be
\label{eq:Uproof}
U:=\left(\frac{\mathbf f(a)}{|\mathbf f(a)|}\ \right |\left.\ \frac{\mathbf g(a)}{|\mathbf g(a)|}\right)
\ee
is an orthogonal matrix, $UU^\top=I$.
Here, we denote $|\mathbf v|:=\sqrt{\mathbf v^\top\mathbf v}$ for a column vector $\mathbf v\in\C^2$.
Introducing 
\be
\label{eq:kappaproof}
\kappa:=\frac{|\mathbf f(a)|\cdot|\mathbf g(a)|}{1-M_{\s,s}(a,a)},
\ee
we have
\be
W_Y(a)=U\begin{pmatrix}0 & \kappa \\ 0 & 0 \end{pmatrix}U^\top.
\ee
Using that $\det Y(z)=1$ identically in $z$ (Remark~\ref{remark:dRHp}) and \eqref{eq:LaurentY},
\begin{align}
\nonumber
\frac 1{c_Y(a)}=\frac{\det Y(z)}{\det C_Y(a)}&=\det\biggl(\frac{1}{z-a}U\begin{pmatrix}0 & \kappa \\ 0 & 0 \end{pmatrix}U^\top+I+Y_1(a)(z-a)+O\bigl((z-a)^2\bigr)\biggr)
\\
\nonumber
&=\det\biggl(\frac{1}{z-a}\begin{pmatrix}0 & \kappa \\ 0 & 0 \end{pmatrix}+I+U^\top Y_1(a)U(z-a)+O\bigl((z-a)^2\bigr)\biggr)
\\
&=1-\kappa\bigl(U^\top Y_1(a)U\bigr)_{2,1}+O(z-a)=1-\kappa\bigl(U^\top Y_1(a)U\bigr)_{2,1}.
\end{align}
Finally, using~\eqref{eq:Uproof} and~\eqref{eq:kappaproof} we get
\be
\kappa\bigl(U^\top Y_1(a)U\bigr)_{2,1}=\frac{\mathbf g^\top(a)Y_1(a)\mathbf f(a)}{1-M_{\s,s}(a,a)}
\ee
and~\eqref{eq:structuredet} follows.

By multiplying the Laurent expansion of $Y^{-1}$, given in~\eqref{eq:LaurentYinverse}, on the right by that of $Y$, given in~\eqref{eq:LaurentY}, vanishing of terms of order $(z-a)^{-1}$ implies $W_Y(a)\wt C_Y(a)C_Y(a)=\wt C_Y(a)C_Y(a)W_Y(a)$.
In turn, this means that $\wt C_Y(a)C_Y(a)=e_Y(a)I+d_Y(a)W_Y(a)$ for some constants $d_Y(a),e_Y(a)$.
Next, the fact that the constant term is the identity gives
\begin{align}
\nonumber
&\wt C_Y(a)C_Y(a)+\wt Y_1(a) \wt C_Y(a)C_Y(a)W_Y(a)-W_Y(a)\wt C_Y(a)C_Y(a)Y_1=I
\\
&\quad \Rightarrow (e_Y(a)-1)I+\bigl(d_Y(a)I+e_Y(a)\wt Y_1(a) \bigr)W_Y(a)=e_Y(a)W_Y(a)Y_1(a).
\end{align}
Multiplying the last relation by~$\mathbf f^\top(a)$ on the left and by~$\mathbf f(a)$ on the right, and combining with~\eqref{eq:structuredet}, we obtain~$e_Y(a)=c_Y(a)$, and so also \eqref{eq:structureCC} is proved.
\end{proof}

Using~\cite[Theorem~1.1]{Borodin}, we immediately obtain the following result.

\begin{theorem}
\label{thm:Borodin}
Let $s\in\Z'$ be such that $Q_\s(L,s)>0$, so that $1-\op M_s$ is invertible.
Then, the discrete RH problem for $Y$ has a unique solution and the resolvent operator $\op R_s:=(1-\op M_s)^{-1}-1$ is induced by the kernel
\be
\label{eq:resolventkernel}
R_s(a,b)=\frac{\wt{\mathbf f}^\top (a)Y^\top(a)Y^{-\top}(b)\wt{\mathbf g}(b)}{a-b},
\quad
R_s(a,a)=\frac{M_{s}(a,a)}{1-M_{s}(a,a)}+\wt{\mathbf g}^\top(a) Y^{-1}(a)Y'(a)\wt{\mathbf f}(a),
\ee
for $a,b\in\Z'$, $a\not=b$, where
\be
\label{eq:FG}
\wt{\mathbf f}(a):=\frac{\mathbf f(a)}{1-M_{s}(a,a)},\qquad \wt{\mathbf g}(a):=\frac{\mathbf g(a)}{1-M_{s}(a,a)}.
\ee
\end{theorem}

Thanks to this result we can prove the following variational formulas for~$Q_\s$.

\begin{theorem}
\label{thm:Qprime}
For all $L>0$ and all $s\in\Z'$ such that $Q_\s(L,s)>0$, we have
\be
\label{eq:dlogQ}
\frac{Q_\s(L,s-1)}{Q_\s(L,s)}-1=\beta(L,s),
\qquad\qquad
\frac{\pa}{\pa L}\log Q_\s(L,s)=-\frac{2\a(L,s)}L,
\ee
where $\a(L,s)$ and $\b(L,s)$ are defined in \eqref{eq:Yinfty}.
\end{theorem}

\begin{proof}
We start with the first equation in~\eqref{eq:dlogQ}.
By~\eqref{eq:lemma2} and~$(1-\op M_s)^{-1}=1+\op R_s$, we have
\be
\label{eq:start}
\frac{Q_\s(L,s-1)}{Q_\s(L,s)}-1=\tr\left((1-\op M_s)^{-1}\op N_s\right)
=\sum_{a,b\in\Z'}\J_{a-\frac 12}\J_{b-\frac 12}\sqrt{\wt\s}(a)\sqrt{\wt\s}(b)\left(\delta_{a,b}+R_{s}(a,b)\right),
\ee
where $\wt\s(a):=\s(a-s-\tfrac 12)$ and $R_{s}(a,b)$ is explicitly given in \eqref{eq:resolventkernel}, and, throughout this proof, we omit the argument~$2L$ of the Bessel functions.
We start by computing the part of the sum that comes from $a\not=b$; denoting $\Delta=\{(a,a):\ a\in\Z'\}$, this is
\be
\label{eq:startoffdiagonal}
\sum_{a,b\in\Z'\setminus\Delta}
\frac{\J_{a-\frac 12}\sqrt{\wt\s}(a)}{1-M_{s}(a,a)}
\frac{\mathbf f^\top(a)Y^\top(a)Y^{-\top}(b)\mathbf g(b)}{a-b}\frac{\J_{b-\frac 12}\sqrt{\wt\s}(b)}{1-M_{s}(b,b)}
=\sum_{a,b\in\Z'\setminus\Delta}\res{z=a}\res{w=b}\frac{\bs\varphi^\top(z)\bs\psi(w)}{z-b}\d w\d z,
\ee
where we introduce the meromorphic vector functions
\be
\label{eq:varphipsi}
\bs\varphi(z):=Y(z)\begin{pmatrix}
0 \\ -1
\end{pmatrix},\qquad
\bs\psi(w):=Y^{-\top}(w)\begin{pmatrix}
-1 \\ 0
\end{pmatrix}.
\ee
Indeed, condition~(b) in the discrete RH problem for $Y$ implies that
\be
\label{eq:Yres}
\res{z=a}Y(z)\d z=\frac{Y(a)\mathbf f(a)\mathbf g^\top(a)}{1-M_s(a,a)},
\ee
yielding
\be
\res{z=a}\bs\varphi(z)\d z=Y(a)\mathbf f(a)\frac{\J_{a-\frac 12}\sqrt{\wt\s}(a)}{1-M_s(a,a)},\qquad
\res{w=b}\bs\psi(w)\d w=Y^{-\top}(b)\mathbf g(b)\frac{\J_{b-\frac 12}\sqrt{\wt\s}(b)}{1-M_s(b,b)}.
\ee
Using condition~(c) in the discrete RH problem for $Y$, we can represent $\bs\psi$ by its (infinite) partial fraction expansion (see Lemma \ref{lemma:partialfraction}), namely
\be
\label{eq:psipartial}
\bs\psi(z)=\begin{pmatrix}-1 \\ 0 \end{pmatrix}+\sum_{b\in\Z'}\frac{\res{w=b}\bs\psi(w)\d w}{z-b}.
\ee
Hence we can rewrite \eqref{eq:startoffdiagonal} as
\begin{multline}
\sum_{a\in\Z'}\res{z=a}\bs\varphi^\top(z)\left[\bs\psi(z)+\begin{pmatrix}1 \\ 0 \end{pmatrix}-\frac{\res{w=a}\bs\psi(w)\d w}{z-a}\right]\d z
\\
\label{eq:prima}
=\sum_{a\in\Z'}\res{z=a}\bs\varphi^\top(z)\begin{pmatrix}1 \\ 0 \end{pmatrix}\d z-\sum_{a\in\Z'}\res{z=a}\res{w=a}\frac{\bs\varphi^\top(z)\bs\psi(w)}{z-a}\d w\d z,
\end{multline}
where we use that $\bs\varphi^\top(z)\bs\psi(z)=0$.
For the first term in \eqref{eq:prima} we appeal to Cauchy theorem to write the sum as a formal residue at $z=\infty$;
\be
\label{eq:seconda}
\sum_{a\in\Z'}\res{z=a}\bs\varphi^\top(z)\begin{pmatrix}1 \\ 0 \end{pmatrix}\d z
=\lim_{n\to+\infty}\frac1{2\pi\i}\oint_{|z|=n}\bs\varphi^\top(z)\d z\begin{pmatrix}1 \\ 0 \end{pmatrix}=\b(L,s),
\ee
where $\b(L,s)$ is introduced in \eqref{eq:Yinfty}.
Using the Laurent expansion \eqref{eq:LaurentY} and \eqref{eq:structureCC}, we compute the second part in \eqref{eq:prima} as
\be
\label{eq:terza}
-\sum_{a\in\Z'}\res{z=a}\res{w=a}\frac{\bs\varphi^\top(z)\bs\psi(w)}{z-a}\d w\d z=-\sum_{a\in\Z'}\frac{c_Y(a)\J_{a-\frac 12}^2\,\wt\s(a)}{1-M_s(a,a)},
\ee
where $c_Y(a):=\det C_Y(a)$.
We now compute the terms in \eqref{eq:start} coming from the diagonal $\Delta\subset\Z'\times\Z'$; this contribution is, using \eqref{eq:resolventkernel} and Lemma~\ref{lemma:structure},
\begin{align}
\nonumber
&\sum_{a\in\Z'}\J_{a-\frac 12}^2\,\wt\s(a)\left(1+\frac{M_s(a,a)}{1-M_s(a,a)}+\wt{\mathbf g}^\top(a)Y^{-1}(a)Y'(a)\wt{\mathbf f}(a)\right)
\\
\nonumber
&\qquad\qquad\qquad
=\sum_{a\in\Z'}\J_{a-\frac 12}^2\,\wt\s(a)\left(\frac{1}{1-M_s(a,a)}+\wt{\mathbf g}^\top(a)\wt C_Y(a) C_Y(a)Y_1(a)\wt{\mathbf f}(a)\right)
\\
\nonumber
&\qquad\qquad\qquad
=\sum_{a\in\Z'}\J_{a-\frac 12}^2\,\wt\s(a)\left(\frac{1}{1-M_s(a,a)}+\frac{c_Y(a)}{1-M_s(a,a)}\frac{\mathbf g^\top(a)Y_1(a)\mathbf f(a)}{1-M_s(a,a)}\right)
\\
\nonumber
&\qquad\qquad\qquad
=\sum_{a\in\Z'}\J_{a-\frac 12}^2\,\wt\s(a)\left(\frac{1}{1-M_s(a,a)}+\frac{c_Y(a)-1}{1-M_s(a,a)}\right)
\\
\label{eq:quarta}
&\qquad\qquad\qquad
=\sum_{a\in\Z'}\frac{c_Y(a)\J_{a-\frac 12}^2\,\wt\s(a)}{1-M_s(a,a)}.
\end{align}
The proof of the first equation in~\eqref{eq:dlogQ} is obtained by combining~\eqref{eq:startoffdiagonal},~\eqref{eq:prima}--\eqref{eq:quarta}.

The proof of the second equation in~\eqref{eq:dlogQ} is similar.
We have
\begin{align}
\nonumber
\frac{\pa}{\pa L}\log Q_\s(L,s)&=-\tr\left((1-\op M_s)^{-1}\frac{\pa\op M_s}{\pa L}\right)
\\
\nonumber
&=-\sum_{a,b\in\Z'}\biggl(\J_{a-\frac 12}\J_{b+\frac 12}+\J_{a+\frac 12}\J_{b-\frac 12}\biggr)\sqrt{\wt \s}(a)\sqrt{\wt\s}(b)\left(\delta_{a,b}+R_s(a,b)\right)
\\
\label{eq:start2}
&=-2\sum_{a,b\in\Z'}\J_{a+\frac 12}\J_{b-\frac 12}\sqrt{\wt \s}(a)\sqrt{\wt\s}(b)\left(\delta_{a,b}+R_s(a,b)\right),
\end{align}
where we use the identity~$(1-\op M_s)^{-1}=1+\op R_s$, the symmetry~$R_s(b,a)=R_s(a,b)$, and we compute $\pa M_s/\pa L$ using
\be
\pa_L\begin{pmatrix}
\J_{a+\frac 12} (2L)\\ L\J_{a-\frac 12}(2L)
\end{pmatrix}
=\begin{pmatrix}
 \frac{a-\frac{1}{2}}{L} & -\frac{2}{L} \\
 2 L & -\frac{a-\frac{1}{2}}{L}
\end{pmatrix}
\begin{pmatrix}
\J_{a+\frac 12}(2L) \\ L\J_{a-\frac 12}(2L)
\end{pmatrix}.
\ee
As before, we start by computing the part of the sum that comes from $a\not=b$; denoting $\Delta=\{(a,a):\ a\in\Z'\}$, this contribution to \eqref{eq:start2} is
\be
\label{eq:startoffdiagonal2}
-2\sum_{a,b\in\Z'\setminus\Delta}
\frac{\J_{a+\frac 12}\sqrt{\wt\s}(a)}{1-M_s(a,a)}
\frac{\mathbf f^\top(a)Y^\top(a)Y^{-\top}(b)\mathbf g(b)}{a-b}\frac{\J_{b-\frac 12}\sqrt{\wt\s}(b)}{1-M_s(b,b)}
=-2\sum_{a,b\in\Z'\setminus\Delta}\res{z=a}\res{w=b}\frac{\bs\omega^\top(z)\bs\psi(w)}{z-b}\d w\d z,
\ee
where we introduce the meromorphic vector functions $\bs\psi$, as in \eqref{eq:varphipsi}, and
\be
\label{eq:omega}
\bs\omega(z):=Y(z)\begin{pmatrix}
1/L \\ 0
\end{pmatrix}
,\qquad\quad
\res{z=a}\bs\omega(z)\d z=Y(a)\mathbf f(a)\frac{\J_{a+\frac 12}\sqrt{\wt\s}(a)}{1-M_s(a,a)},\quad a\in\Z',
\ee
the last equality stemming from~\eqref{eq:Yres}.
Thanks to~\eqref{eq:psipartial}, we rewrite~\eqref{eq:startoffdiagonal2} as
\begin{multline}
-2\sum_{a\in\Z'}\res{z=a}\bs\omega^\top(z)\left[\bs\psi(z)+\begin{pmatrix}1 \\ 0 \end{pmatrix}-\frac{\res{w=a}\bs\psi(w)\d w}{z-a}\right]\d z
\\
\label{eq:prima2}
=-2\sum_{a\in\Z'}\res{z=a}\bs\omega^\top(z)\begin{pmatrix}1 \\ 0 \end{pmatrix}\d z+2\sum_{a\in\Z'}\res{z=a}\res{w=a}\frac{\bs\omega^\top(z)\bs\psi(w)}{z-a}\d w\d z,
\end{multline}
where we use that $\bs\omega^\top(z)\bs\psi(z)$ is regular at $\Z'$.
Again, the first term is a formal residue at $z=\infty$:
\be
\label{eq:seconda2}
-2\sum_{a\in\Z'}\res{z=a}\bs\omega^\top(z)\begin{pmatrix}1 \\ 0 \end{pmatrix}\d z
=-2\lim_{n\to+\infty}\frac1{2\pi\i}\oint_{|z|=n}\bs\omega^\top(z)\d z\begin{pmatrix}1 \\ 0 \end{pmatrix}=-\frac{2\a(L,s)}L,
\ee
where $\a(L,s)$ is introduced in \eqref{eq:Yinfty}.
Using the Laurent expansion \eqref{eq:LaurentY} and \eqref{eq:structureCC}, we compute the second part in \eqref{eq:prima2} as
\be
\label{eq:terza2}
2\sum_{a\in\Z'}\res{z=a}\res{w=a}\frac{\bs\omega^\top(z)\bs\psi(w)}{z-a}\d w\d z=2\sum_{a\in\Z'}\frac{c_Y(a)\J_{a-\frac 12}\J_{a+\frac 12}\wt\s(a)}{1-M_s(a,a)},
\ee
where, as before, $c_Y(a):=\det C_Y(a)$.
With a computation completely analogous to \eqref{eq:quarta} we compute the terms in~\eqref{eq:start2} coming from the diagonal $\Delta\subset\Z'\times\Z'$ as
\be
\label{eq:quarta2}
-2\sum_{a\in\Z'}\J_{a-\frac 12}\J_{a+\frac 12}\wt\s(a)\biggl(1+\frac{M_s(a,a)}{1-M_s(a,a)}+\wt{\mathbf g}^\top(a)Y^{-1}(a)Y'(a)\wt{\mathbf f}(a)\biggr)
=-2\sum_{a\in\Z'}\frac{c_Y(a)\J_{a-\frac 12}\J_{a+\frac 12}\wt\s(a)}{1-M_s(a,a)}.
\ee
The proof of the second equation in~\eqref{eq:dlogQ} is complete by combining~\eqref{eq:startoffdiagonal2},~\eqref{eq:prima2}--\eqref{eq:quarta2}.
\end{proof}

\section{Proof of Theorem~\texorpdfstring{\ref{thm:1}}{I}}\label{sec:thm1}

Throughout this section we shall assume that $s\in\Z'$ is such that $Q_\s(L,s)>0$. In particular (Theorem~\ref{thm:Borodin}), the matrix $Y(z)$ introduced in the last section exists and is unique.

\subsection{Dressing}
We proceed to a \emph{dressing} of the discrete RH problem for~$Y$, mimicking a common technique for continuous RH problems, see, e.g.,~\cite{KH}.
Introduce the following entire matrix function of $z$:
\be
\Phi(z) := \left(\begin{array}{cc}
\J_{z-\frac 12}(2L) &  \i\pi\mathrm H_{z-\frac 12}^{(1)}(2L) \\
L\J_{z+\frac 12}(2L) & \i\pi L\mathrm H_{z+\frac 12}^{(1)}(2L) 
\end{array}\right),
\ee
where $\mathrm H_k^{(1)}(2L)$ is the the Hankel function of the first kind of order~$k$ and argument~$2L$~\cite{DLMF}.
The vectors $\mathbf f$ and $\mathbf g$ in~\eqref{eq:fg} can be expressed as
\be
\label{eq:fgPhi}
\mathbf f(a)=\sqrt\s(a-s-\tfrac 12)\Phi(a)\begin{pmatrix}
1 \\ 0
\end{pmatrix},\qquad
\mathbf g(a)=\sqrt\s(a-s-\tfrac 12)\Phi^{-\top}(a)\begin{pmatrix}
0 \\ -1
\end{pmatrix}.
\ee
Moreover, we have $\det\Phi(z)=1$ identically in $z$~\cite[eq.~10.5.3]{DLMF}; thus $\Phi^{-1}(z)$ is also entire in~$z$.
For later convenience, we also note that $\Phi$ satisfies
\be
\label{eq:dPhi}
\Phi(z+1)=\frac 1L \begin{pmatrix} 0 & 1\\ -L^2 & z+\frac 12\end{pmatrix}\Phi(z),
\qquad
\frac{\pa}{\pa L}\Phi(z)=\frac 1L \begin{pmatrix}
z-\frac12 & -2 \\
 2 L^2 & -z+\frac12
\end{pmatrix}\Phi(z),
\ee
as it follows from the identities~\cite[eq.~10.6.1]{DLMF}
\be
\mathrm B_{k+1}(2L)+\mathrm B_{k-1}(2L)=\frac kL\mathrm B_k(2L),
\qquad
\frac{\pa}{\pa L}\mathrm B_k(2L) = \mathrm B_{k-1}(2L)-\mathrm B_{k+1}(2L),
\ee
where~$\mathrm B_k(\cdot)$ is either of the functions~$\mathrm J_k(\cdot),\mathrm H_k^{(1)}(\cdot)$.

In the interest of clarity, let us momentarily restore the dependence $Y(z)=Y_\s(z;L,s)$ and $\Phi(z)=\Phi(z;L)$.
We introduce the matrix $\Psi(z)=\Psi_\s(z;L,s)$ by
\be
\label{eq:defPsi}
\Psi_\s(z;L,s):=Y_\s(z+s+\tfrac 12;L,s)\Phi(z+s+\tfrac 12;L).
\ee
As we shall now prove, $\Psi(z)$ is uniquely characterized by the following conditions.

\subsubsection*{Discrete RH problem for $\Psi$}
\begin{itemize}
\item[(a)] {\em $\Psi(z)$ is a $2\times 2$ matrix-valued meromorphic function of $z$ with simple poles at $\Z'$ only.}
\item[(b)] {\em For all $a\in\Z'$, the function
\be
\Psi_a^\reg(z):=\Psi(z)\left(I-\frac{W_\Psi(a)}{z-a}\right)
\ee
has a removable singularity at $z=a$, where
\be
\label{eq:WPsi}
W_\Psi(a):=\begin{pmatrix}
0 & -\s(a) \\ 0 & 0
\end{pmatrix}
,\qquad a\in\Z'.
\ee
}
\item[(c)] {\em We have $\lim_{n\to+\infty}\sup_{|z|=n}|\Psi(z)\Phi^{-1}(z+s+\tfrac 12)-I|=0$, where the limit is taken over integer values of $n$, $I$ denotes the identity $2\times 2$ matrix and $|\cdot|$ denotes any matrix-norm.}
\end{itemize}

\begin{proof}
The only condition that does not directly follow from the analogous conditions of the discrete RH problem for $Y$, thus deserving a proof, is (b).
For, we need to show that with $W_\Psi(a)$ as given, for all $a\in\Z'$
\be
\Psi_a^\reg(z)=\Psi(z)\left(I-\frac{W_\Psi(a)}{z-a}\right)
\ee
is regular at~$z=a$.
By using that $W_Y^2(a)=0$ and the definition~\eqref{eq:defPsi} of~$\Psi$, this condition is equivalent to regularity at~$z=a$ of
\be
Y_{a+\wh s}^\reg(z+\wh s)\left(I+\frac{W_Y(a+\wh s)}{z-a}\right)\Phi(z+\wh s)\left(I-\frac{W_\Psi(a)}{z-a}\right),
\ee
where we denote $\wh s:=s+\frac12\in\Z$.
Since $Y_{a+\wh s}^\reg(z+\wh s)$ is regular at $z=a$, we only need to prove that
\be
\left(I+\frac{W_Y(a+\wh s)}{z-a}\right)\Phi(z+\wh s)\left(I-\frac{W_\Psi(a)}{z-a}\right)\mbox{ is regular at }z=a.
\ee
To this end we consider the Laurent expansion at $z=a$ of the previous expression, which is
\begin{multline}
-\frac{W_Y(a+\wh s)\Phi(a+\wh s)W_\Psi(a)}{(z-a)^2}
\\
+\frac{W_Y(a+\wh s)\Phi(a+\wh s)-\Phi(a+\wh s)W_\Psi(a)-W_Y(a+\wh s)\Phi'(a+\wh s)W_\Psi(a)}{z-a}+O(1).
\end{multline}
Vanishing of the coefficient of $(z-a)^{-1}$ implies
\be
\label{eq:WPsicomputation}
W_\Psi(a)=\Phi^{-1}(a+\wh s)W_Y(a+\wh s)\Phi(a+\wh s)\biggl(I+\Phi^{-1}(a+\wh s)\Phi'(a+\wh s)\Phi^{-1}(a+\wh s)W_Y(a+\wh s)\Phi(a+\wh s)\biggr)^{-1}.
\ee
Since $W_Y^2=0$, this also implies that the coefficient of $(z-a)^{-2}$ vanishes and that the series is regular. It remains to show that \eqref{eq:WPsicomputation} simplifies to \eqref{eq:WPsi}. To this end we deduce from \eqref{eq:fgPhi} that
\begin{align}
W_Y(a+\wh s)&=\rho(a,s)\Phi(a+\wh s)\begin{pmatrix}0 & 1 \\ 0 & 0 \end{pmatrix}\Phi^{-1}(a+\wh s),\qquad \rho(a,s):=-\frac{\s(a)}{1-M_s(a+\wh s,a+\wh s)},
\end{align}
such that
\be
W_\Psi(a)=\rho(a,s)\begin{pmatrix}
0 & 1 \\ 0 & 0
\end{pmatrix}\left(I+\rho(a,s)\Phi^{-1}(a+\wh s)\Phi'(a+\wh s)\begin{pmatrix}0 & 1 \\ 0 & 0 \end{pmatrix}\right)^{-1}.
\ee
We now observe by a direct computation that
\be
\label{crossref}
I+\rho(a,s)\Phi^{-1}(a+\wh s)\Phi'(a+\wh s)\begin{pmatrix}0 & 1 \\ 0 & 0 \end{pmatrix}=\begin{pmatrix} 1 & \star \\ 0 & \frac{1}{1-M_s(a+\wh s,a+\wh s)}\end{pmatrix},
\ee
where $\star$ denotes a term whose explicit expression is inconsequential in this computation.
The right-hand side of~\eqref{crossref} is is invertible and so we finally get
\be
W_\Psi(a)=\rho(a,s)\begin{pmatrix}
0 & 1 \\ 0 & 0
\end{pmatrix}\begin{pmatrix}
1 & -(1-M_s(a+\wh s,a+\wh s))\star \\ 0 & 1-M_s(a+\wh s,a+\wh s)
\end{pmatrix}=\begin{pmatrix}
0 & -\s(a) \\ 0 & 0
\end{pmatrix},
\ee
as claimed in \eqref{eq:WPsi}.
\end{proof}

\subsection{Lax pair}

The main result achieved by the dressing procedure is that $W_\Psi(a)$ is independent of $s,L$.
This enables us to obtain the following equations.
It is convenient here to restore the full dependence $\Psi(z)=\Psi(z;L,s)$ (omitting anyway the dependence on $\sigma$ to have lighter notations).

\begin{proposition}
\label{prop:eqPsi}
The matrix $\Psi(z;L,s)$ satisfies
\be
\Psi(z;L,s+1)=\wt A(z;L,s)\Psi(z;L,s),\qquad 
\frac{\pa}{\pa L}\Psi(z;L,s)=\wt B(z;L,s)\Psi(z;L,s)
\ee
where
\begin{align}
\label{eq:wtA}
\wt A(z;L,s)&=\frac 1L \begin{pmatrix}
0 & 1+\beta(L,s+1) \\ -L^2-\g(L,s) & z+s+1+\a(L,s)-\a(L,s+1)
\end{pmatrix},
\\
\label{eq:wtB}
\wt B(z;L,s)&=\frac 1L \begin{pmatrix}
z+s & -2\bigl(1+\beta(L,s)\bigr) \\ 
2\bigl(L^2+\g(L,s)\bigr) & -z-s
\end{pmatrix},
\end{align}
with~$\a(L,s),\b(L,s),\g(L,s)$ as in~\eqref{eq:Yinfty}.
\end{proposition}
\begin{proof}
The fact that $W_\Psi(a)$ is independent of $s$ allows us to write
\be
\wt A(z;L,s):=\Psi(z;L,s+1)\Psi^{-1}(z;L,s)=\Psi_a^\reg(z;L,s+1)(\Psi^\reg_a)^{-1}(z;L,s)
\ee
for all $a\in\Z'$.
Hence $\wt A(z;L,s)$ has removable singularities at $z\in\Z'$ by condition (b) in the discrete RH problem for $\Psi$, and so is an entire function of $z$.
Further, due to \eqref{eq:defPsi} we can write
\begin{align}
\nonumber
\wt A(z;L,s)&=Y(z+s+\tfrac 32;L,s+1)\Phi(z+s+\tfrac 32;L)\Phi^{-1}(z+s+\tfrac 12;L)Y^{-1}(z+s+\tfrac 12;L,s)
\\
\label{eq:dressinglemmaeqdiff}
&=Y(z+s+\tfrac 32;L,s+1)\begin{pmatrix}0 & \frac 1L \\ -L & \frac{z+s+1}L\end{pmatrix}Y^{-1}(z+s+\tfrac 12;L,s),
\end{align}
where we use~\eqref{eq:dPhi}.
This identity, together with condition~(c) in the RH problem for~$Y$, shows that~$\wt A(z;L,s)$ grows linearly as~$z\to\infty$; Liouville theorem then implies that $\wt A(z;L,s)$ is a linear function of~$z$, explicitly obtained by the asymptotic relation~\eqref{eq:Yinfty} plugged in~\eqref{eq:dressinglemmaeqdiff}, which gives the claimed formula for~$\wt A(z;L,s)$.

Similarly, $\wt B(z;L,s):=\bigl(\pa_L\Psi(z;L,s)\bigr)\Psi^{-1}(z;L,s)$ is an entire function of $z$ because $\wt B(z;L,s)=\bigl(\pa_L\Psi_a^\reg(z;L,s)\bigr)(\Psi^{\reg}_a)^{-1}(z;L,s)$ for all $a\in\Z'$ hence the singularities at $\Z'$ are removable.
(Here we use again that $W_\psi(a)$ does not depend on $L,s$.)
Moreover, using~\eqref{eq:defPsi} and~\eqref{eq:dPhi}, we obtain
\be\label{eq:dressinglemmaeqdL}
\wt B(z;L,s)=\bigl(\pa_L Y(z+s+\tfrac 12;L,s)\bigr)Y^{-1}(z+s+\tfrac 12;L,s)
+Y(z+s+\tfrac 12;L,s)\begin{pmatrix}
 \frac{z+s}{L} & -\frac{2}{L} \\
 2 L & -\frac{z+s}{L}
\end{pmatrix}Y^{-1}(z+s;L,s).
\ee
Finally, condition (c) in the RH problem for $Y$ shows that $\wt B(z;L,s)$ grows linearly as $z\to\infty$, and by Liouville theorem it coincides with the linear function of $z$ explicitly obtained by the asymptotic relation~\eqref{eq:Yinfty} plugged in~\eqref{eq:dressinglemmaeqdL}, and this gives the claimed formula for $\wt B(z;L,s)$.
\end{proof}

\begin{remark}
Since $\det\Psi(z)=1$ identically in $z$, we must have $\det\wt A(z)=1$ identically in $z$ as well.
Looking at \eqref{eq:wtA}, this implies the relation
\be
\label{eq:relbetagamma}
\bigl(1+\beta(L,s+1)\bigr)\bigl(L^2+\g(L,s)\bigr)=L^2.
\ee
\end{remark}

\begin{proof}[Proof of Theorem~\ref{thm:1}]
By Proposition~\ref{prop:eqPsi} and equations~\eqref{eq:defPsi} and~\eqref{eq:dPhi}, we have
\be
\frac{\pa Y(z;L,s)}{\pa L}Y^{-1}(z;L,s) =\frac 1L Y(z;L,s)\begin{pmatrix}
z-\frac12 & -2 \\
 2 L^2 & -z+\frac12
\end{pmatrix}
Y^{-1}(z;L,s)+\wt B(z-s-\tfrac 12;L,s).
\ee
Consider the asymptotic expansion of this identity as~$z\to\infty$.
Looking at the entry~$(1,1)$ of the coefficient of~$z^{-1}$ we obtain, also using~\eqref{eq:relbetagamma},
\be
\frac{\pa}{\pa L}\a(L,s) = 2L\left(1-\frac{1+\b(L,s)}{1+\b(L,s+1)}\right).
\ee
The proof is completed by using~\eqref{eq:dlogQ}.
\end{proof}

\section{Proof of Theorem~\texorpdfstring{\ref{thm:ic}}{II}}\label{sec:proofthmic}

The discrete RH problem for $Y$ can be described equivalently as a linear equation on $\ell^2(\Z')\otimes\C^{2}$, as we now explain following the classical operator theory for continuous RH problems and the works of Borodin~\cite{BorodinIMRN,Borodin}.

By conditions~(a) and~(c) in the discrete RH problem for~$Y$ and Lemma~\ref{lemma:partialfraction}, we can write the solution in the form
\be
\label{eq:Ypf}
Y(z)=I+\sum_{b\in\Z'}\frac{R_b}{z-b}.
\ee
The matrices $R_a$ must satisfy, by condition~(b) or, equivalently, \eqref{eq:bstandard},
\be
\label{eq:R1}
R_a=W_Y(a)+\sum_{b\in\Z'\setminus\{a\}}\frac{R_bW_Y(a)}{a-b},\qquad a\in\Z'.
\ee
By~\eqref{eq:WY} we can write
\be
\label{eq:fghat}
W_Y(a)=\wh{\mathbf f}(a)\wh{\mathbf g}^\top(a),\qquad \wh{\mathbf f}(a):=\s(a-s-\tfrac 12)\begin{pmatrix}
\J_{a-\frac 12}(2L) \\ L\J_{a+\frac 12}(2L)\end{pmatrix},\quad\wh{\mathbf g}(a):=\frac 1{1-M_s(a,a)}\begin{pmatrix}
L \J_{a+\frac 12}(2L) \\ -\J_{a-\frac 12}(2L)
\end{pmatrix},
\ee
and so equation~\eqref{eq:R1} implies that $R_a$ is a rank one matrix of the form
\be
\label{eq:RdRH}
R_a=\mathbf r(a)\wh{\mathbf g}^\top(a)
\ee
for some $\mathbf r(a)\in\C^2$ (column vector). Since $\wh{\mathbf g}(a)\not=0$ for all $a\in\Z'$,~\eqref{eq:R1} implies that
\be
\mathbf r(a)=\wh{\mathbf f}(a)+\sum_{b\in\Z'\setminus\{a\}}\frac{\mathbf r(b)\wh{\mathbf g}^\top(b)\wh{\mathbf f}(a)}{a-b}.
\ee
Introduce the operator $\mathsf D:\ell^2(\Z')\otimes\C^2\to\ell^2(\Z')\otimes\C^2$ by
\be
\label{eq:D}
\mathsf D:\bigl(\mathbf r(a)\bigr)_{a\in\Z'}\mapsto\bigl((\mathsf D\mathbf r)(a)\bigr)_{a\in\Z'},\qquad
(\mathsf D\mathbf r)(a):=\sum_{b\in\Z'\setminus\{a\}}\frac{\mathbf r(b)\wh{\mathbf g}^\top(b)\wh{\mathbf f}(a)}{a-b},\ \ a\in\Z'.
\ee
It is a well-defined operator on~$\ell^2(\Z')\otimes\C^2$ by~\eqref{eq:1lineasymptotics} and the fact that $M_s(b,b)=\s(b-s-\tfrac 12)K^\Be(b,b)$ is at a bounded distance from~$1$ for all~$b\in\Z'$ by the assumptions on~$\s$.

If $1-\mathsf D$ is invertible, the discrete RH problem admits a solution, constructed via~\eqref{eq:Ypf} and~\eqref{eq:RdRH} with
\be
\label{solvingdRH}
\mathbf r:=(1-\mathsf D)^{-1}\wh{\mathbf f}.
\ee
The convenience of this approach to the discrete RH problem is evident when the operator~$\mathsf D$ is small.
This is the case when $L\to 0$.
For precision's sake, let us fix the norm on $\ell^2(\Z')\otimes\C^2$ to be the one induced by the standard norm on $\ell^2(\Z')$ and the Euclidean norm on $\C^2$.

\begin{proposition}
\label{prop51}
Let $s\in\Z'$ be such that $Q^0_\sigma(s):=\prod_{i=1}^{+\infty}\bigl(1-\sigma(-i-s)\bigr)>0$.
There exists $L_*=L_*(s),c=c(s)>0$ such that $\|\mathsf D\|<cL$ for~$0\leq L<L_*$, where $\|\mathsf D\|$ is the operator norm of $\mathsf D$.
\end{proposition}
\begin{proof}
For~$k\geq 0$, we have the Taylor series $\J_k(2L)=L^k\sum_{j\geq 0}\frac{(-L^2)^j}{j!(k+j)!}=(-1)^k\J_k(2L)$
which implies
\be
\label{eq:BesselsmallL}
\J_k(2L)=\delta_{k,0}+L(\delta_{k,1}-\delta_{k,-1})+L^2\biggl(\frac{\delta_{k,-2}+\delta_{k,2}}2-\delta_{k,0}\biggr)+L^3\biggl(\frac{\delta_{k,3}-\delta_{k,-3}}6-\frac{\delta_{k,1}-\delta_{k,-1}}2\biggr)+\mathrm O(L^4),
\ee
as $L\to 0$, with remainder uniform in $k\in\Z$ because $|\J_k(2L)|=|\J_{-k}(2L)|\leq L^k/k!$ for $k\geq 0$ integer~\cite[eq.~10.14.4]{DLMF}.
In particular, we have the following estimate for $L\to 0$, uniform in $a\in\Z'$,
\be
\label{eq:KdiagsmallL}
K^\Be(a,a)=\sum_{l\in\Z_+'}\J_{a+l}(2L)^2=\mathbf 1_{a<0}+L^2(\delta_{a,\frac 12}-\delta_{a,-\frac 12})+\mathrm O(L^4).
\ee
(The explicit term of order~$L^2$ will be needed later.)
Therefore, for all $a\in\Z'$, we have
\be
1-M_s(a,a)=1-\sigma(a-s-\tfrac 12)K^\Be(a,a)=\begin{cases}
1+\mathrm O(L^2), & a>0,
\\
1-\sigma(a-s-\tfrac 12)+\mathrm O(L^2),&a<0,
\end{cases}
\ee
with remainders uniform in $a\in\Z'$.
As long as we assume $Q_\s^0(s)\not=0$ we have~$\sigma(a-s-\tfrac 12)\not= 1$ for all~$a\in\Z'$ with~$a<0$, and so we can estimate, for $L$ sufficiently small,
\be
\frac 1{1-M_s(a,a)}\leq c_1(s),
\ee
for a constant $c_1(s)$ depending on $s$ only.
Finally, we can estimate the square of the Hilbert--Schmidt norm of $\mathsf D$ (which is an upper bound of the square of the operator norm of $\mathsf D$) as
\be
\sum_{a,b\in\Z',\ a\not=b}\left|\frac{\wh{\mathbf g}^\top(b)\wh{\mathbf f}(a)}{b-a}\right|^2
\leq
L^2c_1(s)\sum_{a,b\in\Z'}\left|\J_{b+\frac 12}(2L)\J_{a-\frac 12}(2L)-\J_{b-\frac 12}(2L)\J_{a+\frac 12}(2L)\right|^2
\leq 2L^2\e^{2L}c_1(s)
\ee
(where we use again~$|\J_k(2L)|=|\J_{-k}(2L)|\leq L^k/k!$ for~$k\geq 0$ integer in the last step) and the proof is complete.
\end{proof}

\begin{corollary}
\label{corollary:series}
Let $s\in\Z'$ be such that $Q^0_\sigma(s):=\prod_{i=1}^{+\infty}\bigl(1-\sigma(-i-s)\bigr)>0$.
There exists $L_*=L_*(s)$ such that the discrete RH problem for $Y$ is solvable for $0\leq L\leq L_*$, and moreover we have
\be
\label{eq:expansionYsmallL}
Y(z;L,s)=Y^{[0]}(z;s)+L Y^{[1]}(z;s)+L^2 Y^{[2]}(z;s)+L^3 Y^{[3]}(z;s)+\mathrm O(L^4),\qquad L\to 0,
\ee
where $Y^{[i]}(z;s)$ are $2\times 2$ matrix-valued meromorphic functions of $z$ independent of $L$.
\end{corollary}
\begin{proof}
By the discussion above, if the operator $1-\mathsf D$ is invertible, the discrete RH problem for $Y$ admits a solution.
By Theorem~\ref{thm:Borodin}, if $1-\mathsf D$ is invertible, then $Q_\s(L,s)>0$.
It is then enough to use Proposition~\ref{prop51} as well as the formula
\be
\label{eq:operatorY}
Y(z)=I+\sum_{b\in\Z'}\frac{\left((1-\mathsf D)^{-1}\wh{\mathbf f}\right)(b)\wh{\mathbf g}^\top(b)}{z-b},
\ee
stemming from~\eqref{eq:Ypf},~\eqref{eq:RdRH}, and~\eqref{solvingdRH}, along with the Neumann series $(1-\mathsf D)^{-1}=\sum_{k\geq 0}\mathsf D^k$.
\end{proof}
\begin{proof}[Proof of Theorem~\ref{thm:ic}]
The proof follows from the above Corollary~\ref{corollary:series} and the following computations.
In the limit $L\to 0$, the Poissonized Plancherel probability measure converges to a delta measure supported on the empty partition. From~\eqref{eq:multstatPP}, we obtain
\be
Q_\s(L,s)\big|_{L=0}=\prod_{i\geq 0}\bigl(1-\sigma(-i-s)\bigr)=:Q^0_\s(s),\qquad s\in\Z'.
\ee
Next, by~\eqref{eq:KdiagsmallL} and~\eqref{eq:BesselsmallL}, we have that, denoting $\wt\s_s(a):=\s(a-s-\tfrac 12)$,
\begin{align}
\nonumber
W_Y(a)&=\frac{\wt\s_s(a)}{1-\wt\s_s(a)\bigl(\mathbf 1_{a<0}+L^2(\delta_{a,\frac 12}-\delta_{a,-\frac 12})+\mathrm O(L^4)\bigr)}
\\
&\qquad\qquad\qquad
\times\left[
\begin{pmatrix}
0 & -\delta_{a,\frac 12} \\ 0 & 0
\end{pmatrix}
+L^2\begin{pmatrix}
-\delta_{a,-\frac 12}+\delta_{a,\frac 12} & -\delta_{a,-\frac 12}+2\delta_{a,\frac 12}-\delta_{a,\frac 32} \\
\delta_{a,-\frac 12} & \delta_{a,-\frac 12}-\delta_{a,\frac 12}
\end{pmatrix}
+\mathrm O(L^4)
\right]
\end{align}
or, equivalently,
\be
\label{seriesWY}
W_Y(a) = W_Y^{[0]}(a)+L^2W_Y^{[1]}(a)+\mathrm O(L^4),
\ee
where
\begin{align}
\nonumber
W_Y^{[0]}(a)&=-\s(-s)\delta_{a,\frac 12}\begin{pmatrix}
0 & 1 \\ 0 & 0
\end{pmatrix},
\qquad W_Y^{[1]}(a)=\delta_{a,-\frac 12}V_{-\frac 12}+\delta_{a,\frac 12}V_{\frac 12}+\delta_{a,\frac 32}V_{\frac 32},
\\
V_{-\frac 12}&=\frac{\s(-s-1)}{1-\s(-s-1)}\begin{pmatrix}
-1 & -1 \\ 1 & 1
\end{pmatrix},\quad
V_{\frac 12}=\s(-s)\begin{pmatrix}
1 & 2-\s(-s)^2 \\ 0 & -1
\end{pmatrix}
,\quad
V_{\frac 32}=\s(-s+1)\begin{pmatrix}
0 & -1 \\ 0 & 0
\end{pmatrix}.
\end{align}
By Corollary~\ref{corollary:series}, we can solve the discrete RH problem order-by-order in~$L$, i.e., we can plug the expansion~\eqref{eq:expansionYsmallL} into the conditions of the discrete RH problem for $Y(z)$.
Due to the parity of the series~\eqref{seriesWY} it is easy to check that the terms $Y^{[1]}$ and $Y^{[3]}$ in~\eqref{eq:expansionYsmallL} vanish.
In particular, the leading term~$Y^{[0]}(z)$ is characterized by the fact that it is analytic in~$\C\setminus\{\tfrac 12\}$, with a simple pole at~$1/2$, and satisfies
\be
\res{z=1/2}Y^{[0]}(z)\d z=\lim_{z\to 1/2}Y^{[0]}(z)W_Y^{[0]}(\tfrac 12),
\ee
as well as~$\sup_{|z|=n}|Y^{[0]}(z)-I|\to 0$ as~$n\to+\infty$ through integer values.
It follows that
\be
Y^{[0]}(z)=I+\frac{W_Y^{[0]}(\tfrac 12)}{z-\frac 12}.
\ee
Similarly, $Y^{[2]}(z)$ is characterized by the fact that it is analytic in~$\C\setminus\{-\tfrac 12,\tfrac 12,\tfrac 32\}$, with simple poles at $\pm 1/2,3/2$, and satisfies
\be
\label{eq:additiveRH1}
\res{z=-1/2}Y^{[2]}(z)\d z=\lim_{z\to -1/2}Y^{[0]}(z)V_{-\frac 12},\qquad
\res{z=3/2}Y^{[2]}(z)\d z=\lim_{z\to 3/2}Y^{[0]}(z)V_{\frac 32},
\ee
and
\be
\label{eq:additiveRH2}
\res{z=1/2}Y^{[2]}(z)\d z=\lim_{z\to 1/2}\left(Y^{[0]}(z)V_{\frac 12}+Y^{[2]}(z)W_Y^{[0]}(\tfrac 12)\right),
\ee
as well as~$\sup_{|z|=n}|Y^{[2]}(z)|\to 0$ as~$n\to+\infty$ through integer values.
The solution is found in the form
\be
Y^{[2]}(z)=\frac{N_{-\frac 12}}{z+\tfrac 12}+\frac{N_{\frac 12}}{z-\tfrac 12}+\frac{N_{\frac 32}}{z-\tfrac 32}.
\ee
The residues $N_{-\frac 12}$ and $N_{\frac 32}$ are found from~\eqref{eq:additiveRH1} as
\begin{align}
N_{-\frac 12}&=\bigl(I-W^{[0]}_Y(\tfrac 12)\bigr)V_{-\frac 12}=\frac{\s(-s-1)}{1-\s(-s-1)}\begin{pmatrix}
\s(-s)-1 & \s(-s)-1 \\ 1 & 1
\end{pmatrix},
\\
N_{\frac 32}&=\bigl(I+W^{[0]}_Y(\tfrac 12)\bigr)V_{\frac 12}=\begin{pmatrix}
0 & -\sigma(-s+1) \\ 0 & 0
\end{pmatrix}
\end{align}
For the residue~$N_{\frac 12}$, we use~\eqref{eq:additiveRH2} to get
\be
\label{then}
N_{\frac 12}=\lim_{z\to 1/2}\frac{W_Y^{[0]}(\tfrac 12)V_{\frac 12}+N_{\frac 12}W_Y^{[0]}(\tfrac 12)}{z-\tfrac 12}+V_{\frac 12}+(N_{-\frac 12}-N_{\frac 32})W_Y^{[0]}(\tfrac 12).
\ee
Existence of the limit implies
\be
W_Y^{[0]}(\tfrac 12)V_{\frac 12}+N_{\frac 12}W_Y^{[0]}=0
\ee
which can be used to show that
\be
\label{structureN}
N_{\frac 12}=\begin{pmatrix}
\s(-s) & \star \\ 0 & \star 
\end{pmatrix}
\ee
The remaining entries, denoted with $\star$, of $N_{\frac 12}$ can be found then by~\eqref{then} but are not needed for the present argument\footnote{It is however important to note that~\eqref{then} is compatible with the structure~\eqref{structureN} of $N_{\frac 12}$, so that we can really solve for these entries and thus fully determine $Y^{[2]}(z)$.}. Indeed, we have shown that
\be
-\frac L2\frac\pa{\pa L}\log Q_\s(L,s)=\alpha(L,s)=L^2\left(N_{-\frac 12}+N_{\frac 12}+N_{\frac 32}\right)_{1,1}+\mathrm O(L^4),\qquad L\to 0_+,
\ee
where we use~\eqref{eq:dlogQ}, hence the proof is completed by the explicit computation
\be
\left(N_{-\frac 12}+N_{\frac 12}+N_{\frac 32}\right)_{1,1}=\frac{\sigma(-s)-\sigma(-s-1)}{1-\sigma(-1-s)}.
\ee
\end{proof}

\section{Proof of Theorem~\texorpdfstring{\ref{thm:2}}{III}}\label{sec:thm2}

Throughout this section we assume $s\in\Z'$ is large enough such that $Q_\s(L,s-1)>0$. Introduce, as in~\eqref{eq:frakthm},
\begin{align}
\label{eq:fraka}
\mathfrak a(L,s):={}&\frac{\sqrt{Q_\s(L,s+1)Q_\s(L,s-1)}}{Q_\s(L,s)}=\sqrt{\frac{1+\b(L,s)}{1+\b(L,s+1)}},
\\
\label{eq:frakb}
\mathfrak b(L,s):={}&\frac{\pa}{\pa L}\log\frac{Q_\s(L,s)}{Q_\s(L,s-1)}=-\frac 2L\bigl(\a(L,s)-\a(L,s-1)\bigr),
\end{align}
where we use~\eqref{eq:dlogQ}, and
\be
\label{eq:Theta}
\Theta(z;L,s):=\begin{pmatrix}
\frac L{1+\b(L,s)} &0 \\ 0 & 1
\end{pmatrix}\Psi(z;L,s).
\ee
The following proposition is a consequence of~\eqref{eq:wtA}--\eqref{eq:wtB} whose proof is a simple computation that we omit.

\begin{proposition}
We have
\be
\label{eq:Lax}
\Theta(z;L,s+1)=A(z;L,s)\Theta(z;L,s),\qquad 
\frac{\pa}{\pa L}\Theta(z;L,s)=B(z;L,s)\Theta(z;L,s)
\ee
where
\be
\label{eq:AB}
A(z;L,s)=\begin{pmatrix}
0 & 1 \\ -\mathfrak a^2(L,s) & \frac{z+s+1}L+\frac{\mathfrak b(L,s+1)}2 
\end{pmatrix},
\qquad
B(z;L,s)=\begin{pmatrix}
\frac{z+s+1}L+\mathfrak b(L,s) & -2 \\ 
2\mathfrak a^2(L,s) & -\frac{z+s}L
\end{pmatrix}.
\ee
In particular,
\be
\label{eq:Thetastructure}
\Theta(z;L,s)=\begin{pmatrix}
\chi(z;L,s-1) & \wt\chi(z;L,s-1)
\\
\chi(z;L,s) & \wt\chi(z;L,s)
\end{pmatrix}
\ee
where~$f(s)=\chi(z;L,s)$ or~$f(s)=\wt\chi(z;L,s)$ are both solutions to
\be
\label{eq:eqf}
f(s+1)+\mathfrak a^2(L,s)f(s-1)=\left(\frac{z+s+1}L+\frac{\mathfrak b(L,s+1)}2\right)f(s).
\ee
\end{proposition}
\begin{remark}
It is worth noting that the compatibility of~\eqref{eq:Lax} is expressed by the identity
\be
B(z;L,s+1)A(z;L,s)-A(z;L,s)B(z;L,s)=\frac{\pa}{\pa L}A(z;L,s).
\ee
Spelling out this equation gives the two relations
\be\label{eq:Toda_ab}
\frac{\pa\mathfrak a(L,s)}{\pa L} = \frac{\mathfrak a(L,s)}2\bigl(\mathfrak b(L,s+1)-\mathfrak b(L,s)\bigr),\qquad
\frac{\pa\mathfrak b(L,s)}{\pa L} = 4\bigl(\mathfrak a^2(L,s)-\mathfrak a^2(L,s-1)\bigr)-\frac{\mathfrak b(L,s)}L.
\ee
The first identity is manifest already by comparing~\eqref{eq:fraka} and~\eqref{eq:frakb}.
Combining it with the second one immediately implies the 2D Toda equation
\be
\frac{\pa^2}{\pa\theta_+\pa\theta_-}q_s(\theta_+,\theta_-)=\e^{q_{s+1}(\theta_+,\theta_-)-q_s(\theta_+,\theta_-)}-\e^{q_s(\theta_+,\theta_-)-q_{s-1}(\theta_+,\theta_-)}
\ee
for~$q_s(\theta_+,\theta_-):=\log\bigl(Q_\s(\sqrt{\theta_+\theta_-},s)/Q_\s(\sqrt{\theta_+\theta_-},s-1)\bigr)$, which is a consequence of Theorem~\ref{thm:1}.
\end{remark}

The matrix~$W_\Psi(a)$ is not constant in~$a$, and so we shall now obtain a difference equation in~$z$ which, in general, has a meromorphic non-rational matrix coefficient.
To this end, it is convenient to work with the gauge transformed matrix~$\Theta(z;L,s)$ introduced in~\eqref{eq:Theta} and with the functions
\be
\label{eq:defvarphi}
\varphi(z;L,s):=\sqrt{\frac {1+\beta(L,s+1)}L}\,\chi(z;L,s).
\ee
In particular, from~\eqref{eq:eqf} we have
\be
\label{eq:eqvarphi}
\mathfrak a(L,s+1)\varphi(z;L,s+1)+\mathfrak a(L,s)\varphi(z;L,s-1)=\left(\frac{z+s+1}L+\frac{\mathfrak b(L,s+1)}2\right)\varphi(z;L,s).
\ee

\begin{proposition}
We have
\be\label{eq:LaxC}
\Theta(z+1;L,s) = C(z;L,s)\Theta(z;L,s),
\ee
where
\begin{align}
\nonumber
C(z)&=\begin{pmatrix}
0&1\\ -\mathfrak a^2(L,s)&\frac{z+s+1}L
\end{pmatrix}
\\
\label{eq:C}
&\qquad+\sum_{l\in\Z'}\frac{\Delta\sigma(l)}{z-l}\begin{pmatrix}
\mathfrak a(L,s)\varphi(l+1;L,s-1)\varphi(l;L,s) & -\varphi(l+1;L,s-1)\varphi(l;L,s-1)
\\ 
\mathfrak a^2(L,s)\varphi(l+1;L,s)\varphi(l;L,s) & -\mathfrak a(L,s)\varphi(l+1;L,s)\varphi(l;L,s-1)
\end{pmatrix}
\end{align}
where $\Delta\s(l):=\s(l+1)-\s(l)$,~$\mathfrak a(L,s)$ is defined in~\eqref{eq:fraka}, and $\varphi(z;L,s)$ is defined in~\eqref{eq:defvarphi}.
\end{proposition}

\begin{proof}
Define~$C(z;L,s):=\Theta(z+1;L,s)\Theta^{-1}(z;L,s)$. Using~\eqref{eq:defPsi},~\eqref{eq:Theta}, and~\eqref{eq:dPhi}, we obtain
\be
\label{eq:definingC}
C(z;L,s)=\begin{pmatrix}
\frac L{1+\b(L,s)} &0 \\ 0 & 1
\end{pmatrix} Y(z+s+\tfrac 32;L,s)\begin{pmatrix}0 & \frac 1L \\ -L & \frac{z+s+1}L\end{pmatrix}Y^{-1}(z+s+\tfrac 12;L,s)\begin{pmatrix}
\frac {1+\b(L,s)}L &0 \\ 0 & 1
\end{pmatrix}.
\ee
By using this relation, it is straightforward to see that
\be
C_-(z;L,s):=C(z;L,s)-\begin{pmatrix}
0&1\\ -\mathfrak a^2(L,s)&\frac{z+s+1}L
\end{pmatrix}=\mathrm O(1/z)
\ee
as $z\to\infty$, in the sense of Lemma~\ref{lemma:partialfraction}, and applying this lemma we get
\be
C_-(z;L,s)=\sum_{l\in\Z'}\frac{\res{w=l}C_-(w)\d w}{z-l}=\sum_{l\in\Z'}\frac{\res{w=l}C(w)\d w}{z-l}.
\ee
We are left with the task of computing these residues.
For all~$l\in\Z'$, using~\eqref{eq:Theta} we have
\begin{align}
\nonumber
C(z)&=
\Theta(z+1)\Theta^{-1}(z)
\\
\nonumber
&=\Theta^\reg_{l+1}(z+1)\biggl(I+\frac{W_\Psi(l+1)}{z-l}\biggr)\biggl(I-\frac{W_\Psi(l)}{z-l}\biggr)(\Theta^\reg_l)^{-1}(z)
\\
&=\Theta^\reg_{l+1}(z+1)(\Theta^\reg_l)^{-1}(z)+\frac 1{z-l}\Theta^\reg_{l+1}(z+1)\bigl(W_\Psi(l+1)-W_\Psi(l)\bigr)(\Theta^\reg_l)^{-1}(z)
\end{align}
where
\be
\Theta^\reg_l(z):=\begin{pmatrix}
\frac L{1+\b(L,s)} &0 \\ 0 & 1
\end{pmatrix}\Psi^\reg_l(z)=\begin{pmatrix}
\frac L{1+\b(L,s)} &0 \\ 0 & 1
\end{pmatrix}\Psi(z)\biggl(I-\frac{W_\Psi(l)}{z-l}\biggr),
\ee
which is regular at $z=l$.
Hence, using~\eqref{eq:WPsi},~\eqref{eq:Thetastructure},~and the fact that $\det\Theta(z;L,s)=\tfrac L{1+\beta(L,s)}$, we compute $\res{w=l}C(w)\d w$ as
\begin{align}
\nonumber
&\Theta^\reg_{l+1}(l+1)\begin{pmatrix}
0 & \sigma(l)-\sigma(l+1) \\  0 & 0
\end{pmatrix}
\bigl(\Theta^\reg_l\bigr)^{-1}(l)
\\
&=\frac{1+\beta(L,s)}{L}\bigl(\sigma(l)-\sigma(l+1)\bigr)\begin{pmatrix}-\chi(l+1;L,s-1)\chi(l;L,s) & \chi(l+1;L,s-1)\chi(l;L,s-1)
\\ 
-\chi(l+1;L,s)\chi(l;L,s) & \chi(l+1;L,s)\chi(l;L,s-1)
\end{pmatrix}
\end{align}
and the proof is complete by using the definition~\eqref{eq:defvarphi}.
\end{proof}

\begin{proof}[Proof of Theorem~\ref{thm:2}]
We first prove~\eqref{eq:frakanonlocal} and~\eqref{eq:frakbnonlocal}.
To this end, let~$C_{-1}$ be the coefficient of~$z^{-1}$ in the asymptotic series for~$C(z;L,s)$ at~$z=\infty$.
On the one hand, using~\eqref{eq:C}, we have
\be
\label{11C}
(C_{-1})_{1,1}=\mathfrak a(L,s)\sum_{l\in\Z'}\Delta\s(l)\varphi(l+1;L,s-1)\varphi(l;L,s).
\ee 
On the other hand, using~\eqref{eq:definingC} and~\eqref{eq:Yinfty} instead,
\be
\label{11Cbis}
(C_{-1})_{1,1}=L\left(1-\frac{1+\b(L,s)}{1+\b(L,s+1)}\right)=L\bigl(1-\mathfrak a^2(L,s)\bigr),
\ee
where we use~\eqref{eq:fraka} in the last equality.
Hence, \eqref{11C} and \eqref{11Cbis} are equal and we get \eqref{eq:frakanonlocal}.
Similarly, using~\eqref{eq:C}, we have
\be
\tr C_{-1}=\mathfrak a(L,s)\sum_{l\in\Z'}\Delta\s(l)\bigl(\varphi(l+1;L,s-1)\varphi(l;L,s)-\varphi(l+1;L,s)\varphi(l;L,s-1)\bigr),
\ee
whilst using~\eqref{eq:definingC} and~\eqref{eq:Yinfty} we have $\tr C_{-1}=\alpha(L,s)/L$, and so
\be
\alpha(L,s)=L\mathfrak a(L,s)\sum_{l\in\Z'}\Delta\s(l)\bigl(\varphi(l+1;L,s-1)\varphi(l;L,s)-\varphi(l+1;L,s)\varphi(l;L,s-1)\bigr).
\ee
Using this expression and~\eqref{eq:eqvarphi} we finally simplify
\begin{align}
\nonumber
\frac{\alpha(L,s+1)-\alpha(L,s)}L&=\sum_{l\in\Z'}\Delta\s(l)\biggl\lbrace\mathfrak a(L,s+1)\bigl[\varphi(l+1;L,s)\varphi(l;L,s+1)-\varphi(l+1;L,s+1)\varphi(l;L,s)\bigr]
\\
\nonumber
&\qquad\qquad\quad-\mathfrak a(L,s)\bigl[\varphi(l+1;L,s-1)\varphi(l;L,s)-\varphi(l+1;L,s)\varphi(l;L,s-1)\bigr]\biggr\rbrace
\\
\nonumber
&=\sum_{l\in\Z'}\Delta\s(l)\biggl\lbrace\varphi(l+1;L,s)\bigl[\mathfrak a(L,s+1)\varphi(l;L,s+1)+\mathfrak a(L,s)\varphi(l;L,s-1)\bigr]
\\
\nonumber
&\qquad\qquad\quad-\varphi(l;L,s)\bigl[\mathfrak a(L,s+1)\varphi(l+1;L,s+1)+\mathfrak a(L,s)\varphi(l+1;L,s-1)\bigr]\biggr\rbrace
\\
&=-\frac 1L\sum_{l\in\Z'}\Delta\s(l)\varphi(l;L,s)\varphi(l+1;L,s).
\end{align}
Next,~\eqref{eq:eqvarphithm} is exactly~\eqref{eq:eqvarphi}, and it remains only to show the asymptotic relation~\eqref{eq:asympvarphi}.
To this end, we first observe that as~$s\to+\infty$ we have $Q_\s(L,s)\to 1$ and so, by~\eqref{eq:dlogQ},
\be
\beta(s)=\frac{Q_\s(L,s-1)}{Q_\s(L,s)}-1\to 0,\quad\mbox{as }s\to+\infty,
\ee
implying, by~\eqref{eq:defvarphi},~\eqref{eq:Theta}, and~\eqref{eq:Thetastructure}, that
\be
\varphi(z;L,s)\sim\frac 1{\sqrt L}\chi(z;L,s),\qquad s\to+\infty.
\ee
It is therefore enough to show that for all $z\in\Z'$ we have
\be
\chi(z;L,s)\sim L\J_{z+s+1}(2L),\qquad s\to+\infty.
\ee
To this end we first write, using~\eqref{eq:defPsi},~\eqref{eq:Theta},~\eqref{eq:Thetastructure}, and ~\eqref{eq:operatorY},
\begin{align}
\nonumber
\left|\frac{\chi(z;L,s)}{L\J_{z+s+1}(2L)}-1\right|&=
\left|\frac 1{L\J_{z+s+1}(2L)}\biggl((Y(z+s+\tfrac 12)-I)\Phi(z+s+\tfrac 12)\biggr)_{2,1}\right|
\\
\label{eq:vaazero}
&\leq\sum_{b\in\Z'}\left|(0,1)\bigl((1-\mathsf D)^{-1}\wh{\mathbf f}\bigr)(b)\frac{\wh{\mathbf g}^\top(b)\Phi(z+s+\tfrac 12)}{(z+s+\tfrac 12-b)L\J_{z+s+1}(2L)}\begin{pmatrix}
1 \\ 0
\end{pmatrix}\right|
\end{align}
where we recall that the operator~$\mathsf D$ is defined in~\eqref{eq:D}.
We need to show that~\eqref{eq:vaazero} vanishes as $s\to+\infty$.
To this end, we first estimate it as follows
\be
\label{eq:vaazero2}
\left|\frac{\chi(z;L,s)}{L\J_{z+s+1}(2L)}-1\right|\leq c\e^L\sum_{b\in\Z'}\left|(0,1)\bigl((1-\mathsf D)^{-1}\wh{\mathbf f}\bigr)(b)\right|
\ee
because we claim that there exists $c>0$ such that for $s$ sufficiently large (depending on $L,z$ only, not on $b$) we have, for all $b\in\Z'$,
\be
\left|\frac{\wh{\mathbf g}^\top(b)\Phi(z+s+\tfrac 12)}{(z+s+\tfrac 12-b)L\J_{z+s+1}(2L)}\begin{pmatrix}1 \\ 0
\end{pmatrix}\right|\leq c\e^L.
\ee
To prove this last assertion, we rewrite
\be
\label{eq:tobebounded}
\frac{\wh{\mathbf g}^\top(b)\Phi(z+s+\tfrac 12)}{(z+s+\tfrac 12-b)L\J_{z+s+1}(2L)}\begin{pmatrix}1 \\ 0
\end{pmatrix}=\frac{K^\Be(z+s+\tfrac 12,b)}{\bigl(1-M_s(b,b)\bigr)L\J_{z+s+1}(2L)}=\sum_{l\in\Z'_+}\frac{\J_{z+s+\frac 12+l}(2L)}{\J_{z+s+1}(2L)}\frac{\J_{b+l}(2L)}{1-M_s(b,b)}.
\ee
Observe that
\be
M_s(b,b)=\sigma(b-s-\tfrac 12) K^\Be(b,b)\leq \begin{cases}
\displaystyle\sup_{l\in\Z',\ l<-s/2}\s(l), & \mbox{if }b<s/2,
\\[20pt]
K^\Be(\lfloor\frac{s+1}2\rfloor+\frac 12,\lfloor\frac{s+1}2\rfloor+\frac 12), & \mbox{if }b>s/2,
\end{cases}
\ee
which implies
\be
\label{estimatediag}
\frac 1{1-M_s(b,b)}=\mathrm O(1),\quad\mbox{as $s\to+\infty$, uniformly in $b\in\Z'$.}
\ee
Next, for~$k$ real and sufficiently large, $\J_k(2L)$ is positive and monotonically decreasing in~$k$, as it follows, for instance, by~\eqref{eq:1lineasymptotics}.
Therefore, we can bound~\eqref{eq:tobebounded}, provided $s$ is sufficiently large,
\be
\left|\frac{\wh{\mathbf g}^\top(b)\Phi(z+s+\tfrac 12)}{(z+s+\tfrac 12-b)L\J_{z+s+1}(2L)}\begin{pmatrix}1 \\ 0
\end{pmatrix}\right|\leq c\sum_{l\in\Z_+'}|\J_{b+l}(2L)|\leq c\sum_{k\in\Z}|\J_k(2L)|\leq 2c\e^L,
\ee
for some $c>0$, where in the last step we use again the inequality~$\J_{\pm k}(2L)\leq L^k/k!$ for all integers~$k\geq 0$. (In~\eqref{eq:vaazero2} we rename $c\mapsto c/2$.)
Next, we claim that
\be
\label{eq:finalclaim}
\left\|\mathsf D \mathbf r\right\|_{\ell^1(\Z')}\leq \frac 12\left\| \mathbf r\right\|_{\ell^1(\Z')}
\ee
provided $s$ is sufficiently large. Postponing for a while the proof of this claim, let us show how to complete the estimate of~\eqref{eq:vaazero2}: note that $\mathsf D$ commutes with multiplying on the left by the vector $(0,1)$, and therefore so does $(1-\mathsf D)^{-1}$, to write, using~\eqref{eq:finalclaim},
\begin{align}
\nonumber
\sum_{b\in\Z'}\left|(0,1)\bigl((1-\mathsf D)^{-1}\wh{\mathbf f}\bigr)(b)\right|&=
\left\|(0,1)(1-\mathsf D)^{-1}\wh{\mathbf f}\right\|_{\ell^1(\Z')}
\\
\nonumber
&\leq 2\left\|(0,1)\wh{\mathbf f}\right\|_{\ell^1(\Z')}.
\\
\nonumber
&=2L\sum_{a\in\Z'}\sigma(a-s-\tfrac 12)|\J_{a+\frac 12}(2L)|
\\
\nonumber
&=2L\sum_{\begin{smallmatrix} a\in\Z',\\ a<s/2\end{smallmatrix}}\sigma(a-s-\tfrac 12)|\J_{a+\frac 12}(2L)|+2L\sum_{\begin{smallmatrix} a\in\Z',\\ a>s/2\end{smallmatrix}}\sigma(a-s-\tfrac 12)|\J_{a+\frac 12}(2L)|
\\
\label{exactly}
&\leq\left(\sup_{l\in\Z',\ l<-\frac s2}\sigma(l)\right)\sum_{k\in\Z}|\J_k(2L)|+\sum_{\begin{smallmatrix} a\in\Z',\\ a>s/2\end{smallmatrix}}\frac{L^{a+\frac 12}}{(a+\tfrac 12)!}=o(1),
\end{align}
as $s\to+\infty$.
Finally, it remains to prove the claim~\eqref{eq:finalclaim}.
To this end, we have
\be
\label{wecanbound}
\sum_{a,b\in\Z',\ a\not=b}\left|\frac{\mathbf r(b)\wh{\mathbf g}^\top(b)\wh{\mathbf f}(a)}{a-b}\right|
\leq\sum_{b\in\Z'}|\mathbf r(b)|\sum_{a\in\Z'}|\wh{\mathbf g}^\top(b)\wh{\mathbf f}(a)|.
\ee
We can bound this quantity, using $|\J_{b\pm\frac 12}(2L)|\leq \e^L$ and $(1-M_s(b,b))^{-1}\leq c$ for $s$ sufficiently large and for all $b\in\Z'$, as we proved in~\eqref{estimatediag}, as
\be
\mbox{\eqref{wecanbound}}\leq c L\e^L \|\mathbf r(b)\|_{\ell^1(\Z')}\sum_{a\in\Z'}\sigma(a-s-\tfrac 12)\left(|\J_{a+\frac 12}(2L)|+|\J_{a-\frac 12}(2L)|\right)
\ee
and we can bound the last sum over $a$ exactly as in~\eqref{exactly} by splitting it for $a<s/2$ and $a>s/2$, to obtain 
\be
\sum_{a\in\Z'}\sigma(a-s-\tfrac 12)\left(|\J_{a+\frac 12}(2L)|+|\J_{a-\frac 12}(2L)|\right)=o(1),\qquad s\to+\infty,
\ee
such that, indeed, for $s$ sufficiently large we have~\eqref{eq:finalclaim}.
\end{proof}

\subsection{Connection with the discrete Painlev\'e II equation}\label{sec:dPII}

Let us now consider, more specifically, the case $\sigma = \mathbf 1_{\mathbb Z'_+}$, studied in depth by Borodin~\cite{Borodin}.
In this case, \eqref{eq:C} reduces to
\begin{equation}
\label{eq:CBorodin}
	C(z) = \begin{pmatrix}
			0 & 1 \\
			-\mathfrak a^2(L,s) & \frac{z + s + 1}L
		\end{pmatrix} + \frac{1}{z + \frac{1}2}\begin{pmatrix}
									\mathfrak a(L,s)\varphi_+(L,s-1)\varphi_-(L,s) & -\varphi_+(L,s-1)\varphi_-(L,s-1) \\
									\mathfrak a^2(L,s)\varphi_+(L,s)\varphi_-(L,s) & -\mathfrak a(L,s)\varphi_+(L,s)\varphi_-(L,s-1)
								\end{pmatrix}
\end{equation}
where we denoted, for sake of brevity, $\varphi_{\pm}(L,s) = \varphi(\pm 1/2;L,s)$.
In this case, the compatibility conditions between the Lax equations \eqref{eq:Lax} and \eqref{eq:LaxC} greatly simplify, and we recover the well known relations between the discrete Bessel kernel, the discrete Painlev\'e II and the modified Volterra equation (see \cite{Borodin} and also \cite{AvMToda,Hisakado}). 

Let us start by noting that the identities~\eqref{eq:frakanonlocal} and~\eqref{eq:frakbnonlocal} reduce to
\be
\mathfrak{b}(L,s+1)=\frac 2L \varphi_+(L,s)\varphi_-(L,s),\quad
L\bigl(\mathfrak a^{-1}(L,s)-\mathfrak a(L,s)\bigr)=\varphi_+(L,s-1)\varphi_-(L,s).
\ee
Taking the ratio of these gives
\be
\label{eq:last}
\varphi_+(L,s) = \frac{\mathfrak{a}(L,s)\mathfrak{b}(L,s+1)}{2(1 - \mathfrak{a}^2(s))}\varphi_+(L,s-1).
\ee
We then expand the determinant of $C(z;L,s)$, which we know to be equal to $1$, around $z = \infty$.
Using~\eqref{eq:last}, the term of order $z^{-1}$ yields
\be
\label{eq:bor1}
L\mathfrak{a}^2(L,s)(\mathfrak{b}(L,s+1) + \mathfrak{b}(L,s)) = (2s + 1)(1 - \mathfrak{a}^2(L,s)).
\ee
Next, let us consider the compatibility condition between the first equation in \eqref{eq:Lax} and \eqref{eq:LaxC}
\be
	A(z+1;L,s)C(z;L,s) - C(z;L,s+1)A(z;L,z) = 0.
\ee
Inspecting the entry~$(2,2)$ of this condition, once written in terms of $\mathfrak{b}^2(L,s+1), \mathfrak{a}(L,s+1)$ and $\mathfrak{a}(L,s)$ (using the equations obtained before), yields
\be
\label{eq:bor2}
\mathfrak{b}^2(L,s+1) = 4(1 - \mathfrak{a}^2(L,s))(1 - \mathfrak{a}^2(L,s+1)).
\ee
Equations~\eqref{eq:bor1} and~\eqref{eq:bor2} are the same (up to a change of variable) as equations (3.9), (3.10) in~\cite{Borodin} and lead to an expression of the Fredholm determinants $Q_{\mathbf 1_{\mathbb Z'_+}}(L,s)$ in terms of a discrete recursion known as the \emph{discrete Painlev\'e II equation}, eq.~\eqref{eq:dPII} below.
For the reader's convenience, we explain here how to derive it, closely following~\cite{Borodin}.
Note, however, that the Lax pair used to obtain~\eqref{eq:bor1} and~\eqref{eq:bor2} is not the same as the one in op.~cit.

\begin{proposition}[cf. Borodin,~\cite{Borodin}]
Let $v(L,s)$, for $s\in\Z'$ with $s\geq -\tfrac 12$, be the sequence of functions defined by the second order recursion
\be\label{eq:dPII}
	v(L,s+1) + v(L,s-1) = \frac{(s + \frac{1}2)v(L,s)}{L(v^2(L,s) - 1)}
\ee
with initial conditions $v(L,-\tfrac 12) = 1$, $v(L,\tfrac 12) = -\mathrm I_1(2L)/{\mathrm I_0(2L)}$, where $\mathrm I_k(2L)$ is defined in~\eqref{eq:introToeplitz}.
Then, for all $s\in\Z'$ satisfying $s\geq-\tfrac 12$, 
\be\label{eq:detv}
\frac{Q_{\mathbf 1_{\mathbb Z'_+}}(L,s+1)Q_{\mathbf 1_{\mathbb Z'_+}}(L,s-1)}{Q^2_{\mathbf 1_{\mathbb Z'_+}}(L,s)} = 1 - v^2(L,s).
\ee
Moreover, the functions $v(L,s)$ satisfy the modified Volterra equation
\begin{equation}\label{eq:modifiedVolterra}
	\frac{\partial}{\partial L}v(L,s) = \left(1 - v^2(L,s)\right)\left(v(L,s+1) - v(L,s-1)\right).
\end{equation} 
\end{proposition}
\begin{proof}
	We start by defining $v^2(L,-\tfrac 12) = 1$ (which satisfies~\eqref{eq:detv}) and then recursively 
\be
	v(L,s+1) := -\mathfrak{b}(L,s+1)v^{-1}(L,s).
\ee
	Using~\eqref{eq:bor2}, we have $v^2(L,s) = 1 - \mathfrak{a}^2(L,s)$ for all $s\geq-\tfrac 12$.
	We can now write~\eqref{eq:bor1} just in terms of the functions $v(L,s)$, and in this way we obtain~\eqref{eq:dPII}. As for \eqref{eq:detv}, it comes from the equality $v^2(L,s) = 1 - \mathfrak{a}^2(L,s)$ combined with~\eqref{eq:fraka}.
	Finally, the initial condition for~$v(L,1/2)$ can be deduced from the recursive definition  $v(L,s+1) = -\mathfrak{b}(L,s)v^{-1}(L,s)$ combined with \eqref{eq:frakb} and the fact that $Q_{\mathbf 1_{\mathbb Z_+'}}(-1/2) = {\rm e}^{-L^2}, \quad Q_{\mathbf 1_{\mathbb Z_+'}}(1/2) = {\rm e}^{-L^2}{\mathrm I}_0(2L)$, see \eqref{eq:introToeplitz}.
	Finally, the modified Volterra equation~\eqref{eq:modifiedVolterra} is merely a rewriting of~\eqref{eq:Toda_ab} in terms of the functions~$v(s,L)$.
\end{proof}

\paragraph{Acknowledgements.}
We are grateful to Tom Claeys for valuable conversations.
M.C. and G.R. acknowledge support by the European Union Horizon 2020 research and innovation program under the Marie Sk\l odowska-Curie RISE 2017 grant agreement no. 778010 IPaDEGAN and by the IRP ``PIICQ'', funded by the CNRS.
G.R. also acknowledges support by the Fonds de la Recherche Scientifique-FNRS under EOS project O013018F.

\appendix

\section{Partial fraction expansion}\label{app:partialfraction}

\begin{lemma}
\label{lemma:partialfraction}
Let $f(\cdot)$ be a meromorphic function with simple poles at $\Z':=\Z+\tfrac 12$ such that
\be
\max_{|z|=n}|f(z)|\to 0,\qquad\mbox{as $n\to+\infty$ through integer values.}
\ee
Then, for all $z\in\C\setminus\Z'$ we have
\be
f(z)=\sum_{a\in\Z'}\frac{\res{w=a}f(w)\d w}{z-a}.
\ee
\end{lemma}

\begin{proof}
Fix $z\in\C\setminus\Z'$: for all integers $n>|z|$, Cauchy theorem implies that
\be
\frac 1{2\pi\i}\oint_{|w|=n}\frac{f(w)}{z-w}\d w=\sum_{a\in\Z',|a|<n}\frac{\res{w=a}f(w)\d w}{z-a}-f(z).
\ee
As $n\to+\infty$, the left-hand side tends to zero because
\be
\left|\oint_{|w|=n}\frac{f(w)}{z-w}\d w\right|\leq \max_{|w|=n}|f(w)|\oint_{|w|=1}\frac{|\d w|}{|w-z/n|}\to 0,
\ee
as $n\to+\infty$ by assumption, and the proof is complete.
\end{proof}

\end{document}